\documentclass[11pt]{article}

\usepackage{amsthm}
\usepackage{amsmath}
\usepackage{amssymb}
\usepackage{enumerate}
\usepackage{fullpage}

\usepackage{thmtools}
\usepackage{thm-restate}

\newtheorem{theorem}{Theorem}
\newtheorem{corollary}[theorem]{Corollary}
\newtheorem{lemma}[theorem]{Lemma}
\newtheorem{observation}[theorem]{Observation}
\newtheorem{proposition}[theorem]{Proposition}

\newtheorem{claim}[theorem]{Claim}
\newtheorem{fact}[theorem]{Fact}

\newcommand{\setcond}[2]{\left\{#1\: \middle|\: #2\right\}}

\newcommand{\real}{\mathbb{R}}

\newcommand{\naturals}{\mathbb{N}}

\newcommand{\ip}[2]{\langle #1, #2 \rangle}

\newcommand{\mc}[1]{\mathcal{#1}}

\newcommand{\composedwith}{\circ}
\newcommand{\composition}{\odot}
\newcommand{\crit}{{\rm crit}}
\newcommand{\witstar}[1]{\mc{W}^{*}_{#1}}
\newcommand{\wit}[1]{\mc{W}_{#1}}
\newcommand{\alimf}[2]{{#1}^{\lim}(#2)}

\title{Composition limits and separating examples for some Boolean function complexity measures}

\author{Justin Gilmer\thanks{Supported by NSF  grant CCF 083727}\\
Department of Mathematics\\
Rutgers University\\
Piscataway, NJ, USA.\\
\texttt{jmgilmer@math.rutgers.edu}
\and
Michael Saks\thanks{Supported by NSF grants CCF-083727 and CCF-1218711 }\\
Department of Mathematics\\
Rutgers University\\
Piscataway, NJ, USA.\\
\texttt{saks@math.rutgers.edu}
\and
Srikanth Srinivasan\thanks{Work partially done as a Postdoctoral researcher at DIMACS, Rutgers University.}\\
Department of Mathematics\\
IIT Bombay\\
Mumbai, India.\\
\texttt{srikanth@math.iitb.ac.in}
}

\begin{document}

\maketitle

\begin{abstract}
Block sensitivity ($bs(f)$), certificate complexity ($C(f)$) and fractional certificate complexity ($C^*(f)$) are three fundamental
combinatorial measures of complexity of a boolean function $f$. It has
long been  known that $bs(f) \leq C^{\ast}(f) \leq C(f) =O(bs(f)^2)$. We provide an infinite family of examples
for which $C(f)$ grows quadratically in $C^{\ast}(f)$ (and also $bs(f)$) giving optimal separations
between these measures.  Previously the biggest separation known
was $C(f)=C^{\ast}(f)^{\log_{4.5}5}$.  We also give a family of examples for which
$C^{\ast}(f)=\Omega(bs(f)^{3/2})$.

These examples are obtained by composing boolean functions in various ways.  Here the composition $f \circ g$ of
$f$ with $g$ is obtained by substituting for each
variable of $f$ a copy of $g$ on disjoint sets of variables.  To construct and
analyse these examples we systematically investigate the behaviour under function composition
of these measures
and also the sensitivity measure $s(f)$.  The measures $s(f)$, $C(f)$ and $C^{\ast}(f)$
behave nicely under composition: they are submultiplicative (where measure $m$ is
submultiplicative if $m(f \circ g) \leq m(f)m(g)$)
with equality holding under some fairly general conditions.  The measure $bs(f)$ is qualitatively different:
it is not submultiplicative.  This qualitative difference was not noticed in the previous literature and we correct some errors that appeared in previous papers.
We define the composition limit of a measure $m$ at function $f$, $m^{\lim}(f)$ to be the
limit as $k$ grows of $m(f^{(k)})^{1/k}$, where $f^{(k)}$ is the iterated composition of $f$ with itself
$k$-times. For any function $f$ we show that $bs^{\lim}(f) = (C^*)^{\lim}(f)$ and characterize $s^{\lim}(f), (C^*)^{\lim}(f)$, and $C^{\lim}(f)$ in terms of the largest eigenvalue of a certain set of $2\times 2$ matrices associated with $f$. 
\end{abstract}

%\begin{IEEEkeywords}
% Iterated composition; Block sensitivity; Certificate complexity; Fractional Certificate complexity
%\end{IEEEkeywords}

% \fi
%\IEEEpeerreviewmaketitle

\section{Introduction}

\subsection{Measures, critical exponents and iterated limits}

There is a large class of complexity measures for boolean functions
that seek to quantify, for each
function $f$, 
the amount of knowledge about individual variables needed to evaluate  $f$. These
include decision tree complexity and
its randomized and quantum variants, (Fourier) degree, certificate complexity, sensitivity,
and block sensitivity.    The value of such a  measure is at most the number
of variables.  
%We will refer to such measures as {\em variable-based measures}.  
There is a long line of research aimed at bounding one such measure in terms of another.  For 
%variable-based 
measures $a$ and $b$ let us write $a \leq_r b$ if there are constants $C_1,C_2$
such that for every total boolean function $f$, $a(f) \leq C_1 b(f)^r+C_2$. 
For example, the decision tree complexity of $f$, $D(f)$, is at least its degree $deg(f)$ and thus
$deg \leq_1 D$. It is also known \cite{midrijanis} that  $D \leq_3 deg$.
We say that
$a$ is {\em polynomially bounded} by $b$ if $a \leq_r b$ for some $r>0$ and that $a$ and $b$ are
{\em polynomially equivalent} if each  is polynomially bounded by the other.  
The measures mentioned above, with the notable
exception of sensitivity, are known to be polynomially equivalent.  

For a function $f$, the decision tree complexity, degree, certificate complexity, block sensitivity
and sensitivity of $f$ are denoted, respectively, $D(f)$, $deg(f)$, $C(f)$, $bs(f)$ and $s(f)$.  We also
define the fractional certificate complexity $C^{\ast}(f)$ (which is within constant factors of the
randomized certificate complexity defined in \cite{aaronsonqcc}; see Appendix \ref{sec:CstarvsRC}).  
These measures
are defined in Section \ref{sec:preliminaries}; definitions of others may be found
in the survey \cite{buhrmandewolf}.

If measure $a$ is polynomially bounded by $b$ we define the {\em critical exponent
for $b$ relative to  $a$}, $\crit(a,b)$, to be the infimum $r$ such that $a \leq_r b$, which (essentially) gives
the tightest possible upper bound of $a$ as a power of $b$.  In \cite{senssurvey},
there is a table giving the best known upper and lower bounds for the critical exponents
for all pairs from degree, deterministic query complexity, certificate complexity and block sensitivity.
 For example it is known that $\crit(D,C)=2$, while for $crit(D,deg)$ the best bounds known
are  $\log_3(6) \leq \crit(D, deg) \leq 3$. Typically,  
lower bounds on $\crit(a,b)$ (implicitly)  use the following fact:

\begin{proposition}
\label{prop:crit LB}
Let $(f_k:k \geq 1)$ be a sequence of boolean functions for which $b(f_k)$ tends to infinity.
If $\log a(f_k)/\log b(f_k)$ tends to a limit $s$ then $\crit(a,b) \geq s$.  More generally,
\[
\crit(a,b) \geq \lim \inf \frac{\log a(f_k)}{\log b(f_k)}.
\]
\end{proposition}
The proof of this proposition is routine.  Useful lower bounds on $\crit(a,b)$ are obtained  by carefully
selecting the sequence $(f_k)$.  One approach to choosing the sequence is to select some $f_1$ and
define $f_k$ to be the {\em $k$th
iterated composition} of $f_1$, which is defined as follows.
If $f$ and $g$ are boolean functions, respectively, on $n$ and $m$
variables then  $f \composedwith g$ is defined on $nm$ variables split into $n$ blocks of $m$ variables
and is obtained by evaluating $g$ on each block, and then evaluating $f$ on the sequence of $n$ outputs.  
The $k$th iterated composition of $f$ is defined inductively by $f^{(1)}=f$ and 
$f^{(k)}=f \composedwith f^{(k-1)}$ for $k \geq 2$.  We say that a complexity measure $a$ is {\em multiplicative with respect
to function $f$} if $a(f^{(k)}) = a(f)^k$ for all $k \geq 1$.  We say that $a$ is {\em multiplicative} if
for any two functions $f$ and $g$ we have $a(f \composedwith g)=a(f)a(g)$; this condition implies
immediately that $a$ is multiplicative with respect to every function $f$.   As a direct consequence of Proposition
\ref{prop:crit LB} we have:

\begin{proposition}
\label{prop:crit LB 2}
If $a$ and $b$ are complexity measures that are each multiplicative with respect to the function $f$
then $\crit(a,b) \geq \log a(f)/\log b(f)$.
\end{proposition}

For example, the lower bound $\crit(D,deg) \geq \log_3 6$ is obtained by applying Proposition
\ref{prop:crit LB 2} to a specific six variable boolean function $f$ having
$deg(f)=3$ and $D(f)=6$ using the easy fact that the measures $deg$ and $D$ are multiplicative.

For non-multiplicative measures $a$, $b$
one may be able to identify specific functions $f$ such that $a$ and $b$ are each
multiplicative on  $f$ which is enough to use Proposition
\ref{prop:crit LB 2}. 
%Measures $s$,$bs$, $C(f)$ and $C^{\ast}(f)$ are all  
%{\em induced by local measures}.    A local measure $m$ has an additional argument
%an input $x$.  The value of $m$ for$f$ and input $x$ is denoted $m_x(f)$. For $b \in \{0,1\}$, 
%$m_b(f)$ is  
%the maximum of $m_x(f)$ over  $x \in f^{-1}(b)$, and $m(f)=\max(m_0(f),m_1(f))$.
%We also define $m_0(f)$ to be the maximum of $m_x(f)$ over all $x \in f^{-1}(0)$ and
%$m_1(f)$ to be the maximum of $m_x(f)$ over all $x \in f^{-1}(1)$.
While $s,C,C^{\ast}$ are not multiplicative,  each is multiplicative on functions $f$ that
satisfy $m_0(f)=m_1(f)$ (see Section \ref{sec:complexmeasures} for definitions).  
However, this fails for block sensitivity, and this
failure is responsible for some errors in the literature.  
In \cite{aaronsonqcc} it was proposed that a six variable function $f$ given by Paterson (see \cite{bublitz}) could
be used
to obtain a lower bound on the critical exponent of block sensitivity relative to
certificate complexity.  The function $f$ has block sensitivity $4$ and certificate complexity $5$
and in fact satisfies $bs_x(f)=4$  and $C_x(f)=5$ for all inputs $x$.  This was used to deduce
that both block sensitivity and certificate complexity are multiplicative
on  $f$, and therefore $\crit(C,bs) \leq \log_4 5$.  It turns out, however, that
block sensitivity is not multiplicative with respect to $f$. In this case, $bs(f^{(m)})^{1/m}$ tends to $4.5$ rather than $4$
and so the resulting lower bound on $crit(C^{\ast},bs)$
is $\log_{4.5}(5)$ rather than $\log_4 5$.

Proposition \ref{prop:crit LB 2} can be extended to the case that $a$ and $b$ are not
necessarily multiplicative on $f$.  Given any 
%variable-based 
measure $a$
we can define a new 
%variable-based complexity 
measure $a^{\lim}$, called the {\em composition limit of $a$},
 where $a^{\lim}(f)=
\lim\inf a(f^{(k)})^{1/k}$.  If $a$ is multiplicative on $f$ then
$a^{\lim}(f)=a(f)$.   Applying Proposition \ref{prop:crit LB} yields the following
extension of Proposition \ref{prop:crit LB 2}:

\begin{proposition}
\label{prop:crit LB 3}
Suppose $a$ and $b$ are complexity measures and $f$ is a boolean function for which $b^{\lim}(f)>1$.
Then $\crit(a,b) \geq \log a^{\lim}(f)/\log b^{\lim}(f)$.
\end{proposition}

To apply this proposition, we need to analyse $a^{\lim}(f)$ and $b^{\lim}(f)$.
% and depending on the measures and the function, this may not be easy.

%that the critical exponent of certificate complexity relative to
%deterministic query complexity is at least $2$ is shown by the sequence of functions $f_k$,
%on $k^2$ variables divided into $k$ blocks of size $k$, which is 1 if at least one of the blocks
%as all of its variables equal to 1; this function has certificate complexity $k$ and deterministic
%decision tree complexity at least $k^2$.

\subsection{The contributions of this paper}

In this paper we analyse the behaviour of certificate complexity, fractional certificate complexity,
 sensitivity and block
sensitivity under composition.  This enables us to give characterizations
of the composition limits  $s^{\lim}$, $C^{\lim}$, $ (C^*)^{\lim}$ and $bs^{\lim}$. We also obtain new lower bounds
on $\crit(C,bs)$, $\crit(C,C^{\ast})$ and  $\crit(C^{\ast},bs)$; in the first two cases the new lower
bounds are tight.
  
  The paper is organized as follows.

\begin{itemize} 
\item In Section \ref{sec:preliminaries} we give various definitions and technical preliminaries. We introduce a new notion of an assemblage, which provides a common abstraction for the objects underlying the measures $s(f), C(f), bs(f)$, and $C^*(f)$.
\item 
In Section \ref{sec:Main}, we characterize the composition limit of $s(f)$, $C(f)$ and $C^{\ast}(f)$.   For $m \in \{s(f),C(f),C^{\ast}(f)\}$, 
we always have  $\min\{m_0(f),m_1(f)\} \leq m^{\lim}(f) \leq m(f)$.  
We express the composition limit
as the minimum over a certain family of 2 by 2 matrices (determined by the function $f$ and the complexity
measure) of the largest eigenvalue of the matrix.      
\item
In Section \ref{sec:blocktensor}, we consider the composition limit of block sensitivity.  %As mentioned earlier,
%the composition limit for block sensitivity has qualitatively different behavior from $s$, $C$ and $C^{\ast}$.
%In particular we may have $bs^{\lim}(f) > bs(f)$.
We prove Theorem \ref{thm_bs_tensor}
which says that for any boolean function $f$, the composition limit of $bs(f)$ is equal to 
the composition limit of  $C^{\ast}(f)$.   
\item In Section \ref{sec:scottremark}, we discuss the previously mentioned example from \cite{bublitz} and correct the analysis
of $bs^{\lim}(f)$.
\item In Section \ref{sec:examples}, we give improved separations between block sensitivity and  fractional
block sensitivity, and between block sensitivity and certificate complexity. %For any boolean function $f$, $bs(f) \leq C^{\ast}(f) \leq C(f) \leq bs(f)^2$.  Thus $\crit(C^{\ast}, bs) $, $\crit(C,bs)$ and $\crit(C,C^{\ast})$
%are all bounded between 1 and 2. 
We present two distinct examples that give the tight lower bounds $\crit(C,C^{\ast}) \geq 2$ and $\crit(C,bs) \geq 2$ and 
an example that shows $\crit(C^*, bs) \geq 3/2$.
\item In Appendix \ref{sec:CstarvsRC} we prove that fractional certificate complexity is within a constant factor of the
randomized certificate complexity defined in \cite{aaronsonqcc}.
\end{itemize}

Independently, Tal (\cite{taleccc,talitcs}) proved results that have some overlap with our work. He showed that $bs(f)$ is not submultiplicative and proved that $(C^*)^{lim}(f) = bs^{lim}(f)$. He also observed the submultiplicativity of the measures $C(f), C^*(f),$ and $s(f)$. Finally, he showed a lower bound on $\crit(C,C^*)$ of $\log(26)/\log(17)$, which we improve here to the optimal constant $2$.

%Note: A few days after we submitted this paper we learned that the previous week, Avishay
%Tal had posted a closely related paper to ECCC (\cite{taleccc,talitcs}).  In that paper he obtained a number of the same
%technical results about function composition of block sensitivity, certificate complexity and fractional
%block sensitivity that we obtained. In particular, he also obtained the result that $C^{*\lim}(f)=bs^{\lim}(f)$.
%His paper does not give the precise characterizations of $s^{\lim}$, $C^{*\lim}$ and $C^{\lim}$ in terms
%of the maximum eigenvalue over a certain set of 2 by 2 matrices.  Tal's paper contains a counterexample
%to a conjecture of Kalai about restrictions of functions with high fourier degree that we don't have.  
%The separation results for $bs(f)$, $C^{\ast}(f)$ and $C(f)$  we obtained are somewhat stronger.
%While he obtained a bound of $\crit(C^{\ast},C) \geq \log(26)/\log(17)$, we obtained  a tight
%lower bound of $\crit(C^{\ast},C) \geq 2$.  Also, Tal also proved that the iterated composition of a single
%function is not capable of giving a nontrivial lower bound on $\crit(bs,C^{\ast})$ and asked whether
%there is such a lower bound.  Our result that $\crit(bs,C^{\ast}) \geq 3/2$ answers this question though we don't
%know that this bound is tight.
%
%

\section{Preliminaries} \label{sec:preliminaries}

\subsection{Combinatorial objects over an index set $I$}
\label{sec:objects over I}
Let $I$ be an arbitrary finite set, called the {\em index} set.  We will be considering a large number of mathematical objects built relative to $I$.

%If we introduce an object $\alpha$ defined
%with respect to index set $I$ and wish to emphasize the index set we may denote this by %$\alpha_I$.
\begin{itemize}
\item 
A map from $I$ to the nonnegative reals is called a {\em weight function over $I$}.  
 A weight function is said to be {\em $[0,1]$-valued} (respectively {\em  integral}, {\em  boolean})
 if all weights lie in $[0,1]$ (respectively $\mathbb{Z}$,
$\{0,1\}$).  A boolean weight function $w$ corresponds naturally to the subset $w^{-1}(1)$.
For any weight function $w$ over $I$ and $J \subseteq I$ we write $w(J)$ for $\sum_{j \in J} w(j)$ and $|w|$ for $w(I)$.
\item A {\em weight function family over $I$} is a set of weight functions over $I$.  The family is $[0,1]$-valued (respectively integral,
boolean) if weight functions in the family have this property.  A boolean weight function family corresponds in the obvious way to a
collection of subsets (hypergraph) on $I$.   We will use the terms hypergraph and boolean weight function family
interchangeably.
\item
A {\em boolean assignment over $I$}  or, simply, an {\em assignment}  is a map from $I$ to $\{0,1\}$.
\item 
 A {\em boolean function over $I$} is a map from assignments over $I$ to $\{0,1\}$.
%\item A {\em function-assignment} pair over $I$ is a pair $(f,x)$ where $f$ is a boolean function
%and $x$ is an assignment.
%\item
%A {\em partial assignment over $I$} is a mapping $\pi$ from $I$ to
%$\{0,1,*\}$.  Indices assigned to $0$ or $1$ are the {\em fixed indices} of $\pi$
%and the indices assigned to $*$ are the {\em free indices} of $\pi$.  We write $Fixed(\pi)$ and %$Free(\pi)$ for
%the sets of fixed and free indices. An assignment $a$ {\em extends} partial assignment $\pi$
%if $a_i=\pi_i$ for all $i \in Fixed(\pi)$. The set of assignments that extend $\pi$ forms a subcube
%of $\{0,1\}^I$ and is denoted $Ext(\pi)$. 
%If $a$ is an assignment and $S \subseteq I$ we write $a[S]$ for the partial assignment with fixed set $S$
%that agrees with $a$ on $S$.
%\item If $f$ is a boolean function over $I$ and $\pi$ is a partial assignment over $I$ we say that
%$\pi$ is a {\em certificate for $f$} if $f$ restricted to $Ext(\pi)$ is constant.  If that constant
%is $b \in \{0,1\}$ we say that $\pi$ is a $b$-certificate for $f$.
\end{itemize}

%These definitions are standard.   
We now introduce
a few non-standard notions:

\begin{itemize}
\item A {\em selector} is a function on domain $\{0,1\}$.  We typically denote selectors by
vector notation  $\vec{\alpha}=(\alpha^0,\alpha^1)$.
\item An {\em assignment selector} is a selector $\vec{\alpha}=(\alpha^0,\alpha^1)$
where $\alpha^0$ and $\alpha^1$ are boolean assignments over $I$.  
\item An assignment selector $\vec{\alpha}$ is {\em $f$-compatible} for a boolean function $f$ provided that $f(\alpha^0)=0$
and $f(\alpha^1)=1$.
\item A {\em weight function selector} is a selector $\vec{w} = (w^0,w^1)$ where $w^0$ and $w^1$ are weight functions over $I$.
%\item A function-selector pair $(f,\vec{\alpha})$ consists of a boolean function $f$
%and an $f$-compatible selector $\vec{\alpha}.$
%\item A {\em partial assignment selector over $I$} is a pair $\vec{\pi}=(\pi^0,\pi^1)$ of partial assignments.
%\item A partial assignment selector $\vec{\pi}$ is a {\em certificate selector} for the boolean function $f$ if
%$\pi^0$ is a 0-certificate and $\pi^1$ is a 1-certificate.
\end{itemize}

\subsection{Packing and covering in hypergraphs}
\label{sec:PandC}

In the previous section we introduced both hypergraphs over $I$ and weight functions over $I$.  We will also
need to consider weight functions whose domain is $\mc{H}$ (rather than $I$).
For a hypergraph $\mc{H}$ on $I$, we have the following (fairly standard) definitions:

\begin{itemize}
\item  For a weight function $w$ on $I$, a {\em fractional $w$-packing} of $\mc{H}$ is a weight function $\lambda$ on $\mc{H}$
with the property that for each $i \in I$ the sum of $\lambda(E)$ over all $E$ containing $i$ is at most $w(i)$.
If we omit the word {\em fractional} then  $\lambda$ is assumed to be integer valued.   
Given $M\in\naturals$, an \emph{$M$-fold packing} of $\mc{H}$ is an integral $w$-packing for the constant weight function $w(\cdot) \equiv M$.
Thus a 1-fold packing corresponds to a collection of pairwise disjoint edges of $\mc{H}$, and is called
simply a {\em packing}.  
The weight of a (fractional) packing $\lambda$, denoted $|\lambda|$ is the sum of $\lambda(E)$ over all $E \in \mc{H}$.
\item A {\em fractional hitting set} for $\mc{H}$ is a weight function $\beta$ on $I$ satisfying
$\beta(E) \geq 1$ for all $E \in \mc{H}$.  If ``fractional''  is omitted then $\beta$ is assumed to be boolean and so
corresponds to a subset $S$ of $I$ that meets every edge. 
\item $\nu(\mc{H})$, $\nu^w(\mc{H})$, $\nu^M(\mc{H})$ and $\nu^{\ast}(\mc{H})$ denote the maximum weight (size) of a packing of $\mc{H}$, the maximum size of a $w$-packing of $\mc{H}$,
the maximum size of an $M$-fold packing of $\mc{H}$, and the maximum weight of a fractional packing of $\mc{H}$ respectively.
\item $\tau(\mc{H})$ and $\tau^{\ast}(\mc{H})$ denote the size of the smallest hitting set of $\mc{H}$ and the minimum weight of a fractional hitting set for $\mc{H}$ respectively.
\end{itemize}

For a hypergraph $\mc{H}$ we denote by $\partial\mc{H}$ the hypergraph consisting of those edges of $\mc{H}$ that are minimal under inclusion.   It is not hard to see that all of the definitions above for a hypergraph $\mc{H}$ only
depend on $\partial{\mc{H}}$.

The following chain of relations always holds:

\[
\nu(\mc{H}) \leq \nu^{\ast}(\mc{H}) = \tau^{\ast}(\mc{H}) \leq \tau(\mc{H}).
\]
where the inequalities are immediate consequences of the definitions and the equality
follows from the duality theorem of linear programming.

It is also known (see \cite[Chapter 1]{fgt}) that $\tau^\ast(\mc{H}) = \lim_{M\rightarrow\infty} \tau^M(\mc{H})/M = \sup_{M} \tau^M(\mc{H})/M$.

\subsection{Assemblages}

An {\em assemblage $\mc{A}$ over $I$} is a map which associates each 
function-assignment pair $(f,x)$ to a family  $\mc{A}_x(f)$ of weight functions over $I$
that is compact (when viewed as a subset of $\mathbb{R}^I$).
An important special case is when all of the weight functions are $\{0,1\}$-valued, in which case $\mc{A}_x(f)$ can be
viewed as a hypergraph.  Some important examples of assemblages are:

\begin{itemize}
\item {\em The block assemblage $\mc{B}$}.  A {\em block} of $f$ at $x$ is a subset $B$ of $I$ such
that $f(x \oplus B) \neq f(x)$ where $x \oplus B$ is obtained by complementing the bits of $x$
in the positions indexed by $B$. For the block assemblage $\mc{B}$, $\mc{B}_x(f)$ is equal
to the set of blocks of $f$ at $x$.
\item {\em The minblock assemblage $\partial\mc{B}$.}  A {\em min-block} of $f$ at $x$
is a block which is minimal under containment (but not necessarily minimum size). We define $\partial\mc{B}_x(f)$ to be the set of
min-blocks of $f$ at $x$.  %Note that this notation is compatible with the $\partial\mc{H}$
%notation introduced for hypergraphs.
\item {\em The witness assemblage $\wit{}{}$.}  A {\em witness} $w$ of $f$ at $x$ is a hitting
set for $\mc{B}_x(f)$ (equivalently, for $\partial\mc{B}_x(f)$).
For the witness assemblage $\wit{}{}$, $\wit{x}(f)$ is the set of witnesses of $f$ at $x$. Note we view witnesses as boolean valued weight functions over the index set of $f$.
\item {\em The fractional witness assemblage $\witstar{}{}$.} A {\em fractional witness} $w$ of $f$
at $x$ is a fractional hitting set for $\mc{B}_x(f)$, and $\witstar{x}(f)$ is the set of all
fractional witnesses for $f$ at $x$. 
%Note we view fractional witnesses as $[0,1]$-valued weight functions over the index set of $f$. 
Thus a weight function $w$ on $I$ belongs to  $\mc{W}^*_x(f)$ if and only if: 
  \begin{itemize}
  \item $0 \leq w(i) \leq 1$ for each $i \in I$
  \item $w(B) \geq  1$ for each $B \in \partial\mc{B}_x(f)$.
  \end{itemize}

\item {\em The sensitivity assemblage $\Psi$.}  A {\em sensitive index} for $f$ at $x$
is an index $i$ such that $\{i\}$ is a block.  For the sensitivity assemblage $\Psi$, $\Psi_x(f)$ consists of a single boolean
weight function which is 1 on the set of indices that are sensitive for $f$ at $x$ and 0 otherwise.
%We abuse notation and write $\Psi_x(f)$ both
%for the set of sensitive indices and the hypergraph whose only hyperedge
%is the set of sensitive indices.
\end{itemize}

%%%%%%%%%%%%%%%%%%%%%%%%%
\iffalse
We note the following easy fact:

\begin{proposition}
Let $f$ be a boolean function and $x$ an input.  If $W$ is an index subset then
$W$ is a witness for $f$ at $x$ if and only if the restriction $x[W]$ is a certificate
for $f$ at $x$.  
\end{proposition}
\fi

\subsection{Local complexity measures} 
\label{sec:complexmeasures}

 A {\em local complexity
measure} $m$ depends on a function $f$ and an input $x$ to the function.  The value
is written $m_x(f)$ and is read as the $m$-complexity of $f$ at $x$. Given such a local complexity measure
we define: 
\begin{eqnarray*}
m_0(f) & =\ & \max \{m_x(f):x \in f^{-1}(0)\} \\
m_1(f) & = &\max \{m_x(f):x \in f^{-1}(1)\} \\
\vec{m}(f)& = &(m_0(f),m_1(f)) \\
m(f) & = &\max\{m_0(f),m_1(f)\}\\
m^{\lim}(f) &= &\liminf_{k \rightarrow \infty} m(f^{(k)})^{1/k}\\
%m_0^{\lim}(f) & = & \lim\inf_{k \rightarrow \infty} m_0(f^{(k)})^{1/k}\\
%m_1^{\lim}(f) & = & \lim\inf_{k \rightarrow \infty} m_1(f^{(k)})^{1/k}
\end{eqnarray*}

 The measure $m(f)$ is said to be {\em induced by a local complexity measure}. 
 Each of the following (standard) combinatorial measures of complexity of $f$, i.e., {\em certificate complexity}, {\em fractional certificate complexity}, {\em sensitivity} and {\em block sensitivity}, are induced by local complexity measures. The corresponding local measures are defined in the following subsections.

Let $m$ be induced by a local complexity  measure. For a function $f$ and an
$f$-compatible selector $\vec{\alpha}$, we define $\vec{m}_{\vec{\alpha}}(f):=\left(m_{\alpha^0}(f),m_{\alpha^1}(f)\right)$.
Note that $\vec{m}(f) \geq \vec{m}_{\alpha}(f)$  (coordinate-wise) with equality
if and only if $\alpha^0$ maximizes $m_x(f)$ over all $x \in f^{-1}(0)$ and
$\alpha^1$ maximizes $m_x(f)$ over  $x \in f^{-1}(1)$.  In this case we say
that $\vec{\alpha}$ is {\em an $m$-optimal selector for $f$}.

\subsection{Assemblage-based measures}\label{sec:assembmeasures}

Associated to any assemblage  $\mc{A}$ is a local complexity measure $m=m[\mc{A}]$
where $m_x(f)$ is equal to the minimum of $|w|$ over all $w \in \mc{A}_x(f)$.
We say that this complexity measure is {\em induced by assemblage $\mc{A}$}.
In this way we define the following local complexity measures:

\begin{itemize}
\item The certificate complexity of $f$ at $x$, $C_x(f)$ is the minimum of $|w|$ over $w
\in \wit{x}(f)$.
\item The fractional certificate complexity of $f$ at $x$, $C^{\ast}(f)$,  is the minimum of $|w|$
over $w \in \witstar{x}(f)$.
\item The sensitivity of $f$ at $x$, $s_x(f)$, is the number of sensitive indices of $f$ at $x$
which is (trivially) the size of the set in $\Psi_x(f)$.
\end{itemize}

Fix an assemblage $\mc{A}$ with associated local complexity measure $m$ and a boolean function $f$. Let $\vec{\alpha} = (\alpha^0,\alpha^1)$ be an $f$-compatible assignment selector and let $\vec{w}=(w^0,w^1)$ be a weight function selector. We say that $(\vec{\alpha},\vec{w})$ form an {\em $(f,\mc{A})$-compatible pair} if $w^0 \in \mc{A}_{\alpha^0}(f)$ and $w^1 \in \mc{A}_{\alpha^1}(f)$. For such a compatible pair, we say $\vec{w}$ is {\em $\vec{\alpha}$-compatible}.

If $(\vec{\alpha},\vec{w})$ form an $(f,\mc{A})$-compatible pair, $|w^0| = m_{\alpha^0}(f)$, and $|w^1| = m_{\alpha^1}(f)$, then we say $\vec{w}$ is an {\em $m$-optimal selector for $f$ at $\vec{\alpha}$}.

\subsection{Block sensitivity and its variants}
\label{sec:bs_def}

Next we define some local complexity measures related to packings of blocks:

\begin{itemize}
\item $bs_x(f)$, the {\em block sensitivity of $f$ at $x$}, is $\nu(\mc{B}_x(f))$, the size of the maximum packing of blocks. %Note that block sensitivity is not induced by the assemblage $\mc{B}$ in the sense of section \ref{sec:assembmeasures}.
\item $bs^\ast_x(f)$, {\em the fractional block sensitivity of $f$ at $x$}, is $\nu^{\ast}(\mc{B}_x(f))$, the weight
of the maximum fractional packing of blocks.
\item $bs^M_x(f)$, {\em the $M$-fold block sensitivity of $f$ at $x$}, where $M$ is a positive integer,
is $\nu^M(\mc{B}_x(f))$, the weight of the maximum $M$-fold packing of blocks.
\item $bs^w_x(f)$, {\em the $w$-block sensitivity of $f$ at $x$}, where $w$ is a weight function on $I$,
is $\nu^w(\mc{B}_x(f))$, the $w$-block sensitivity of $f$ at $x$.
\end{itemize}

Applying the general inequalities for hypergraph parameters (mentioned in Section \ref{sec:PandC}) we have:
\[ s_x(f) \leq bs_x(f) \leq bs^\ast_x(f) = C^\ast_x(f) \leq C_x(f).
\] 

Also, we have 
\[bs^\ast(f) = \lim_{M\rightarrow\infty} bs^M(f)/M = \sup_M bs^M(f)/M.\]

\subsection{Compositions}
\label{sec:comp}

We will need to define the  composition of various  objects over an index set. For this purpose, it is convenient to represent an index set as the set of leaves of a rooted tree.

\iffalse
%l
\subsubsection{Index sets, indexed trees and their compositions}
\label{subsubsec:indexed trees}

For finite strings $s,t$ we write $s \composedwith t$ for their concatenation.  A set of strings is said to be {\em prefix-free} if
no string is a prefix of another.   In what follows our index sets will always be prefix-free sets of strings,
and we refer to such a set as an {\em index set}.    We allow $I$ to be empty.
If $I$ is an index set, all of whose strings have length 1,
we say that $I$ is an {\em elementary} index set, otherwise we say that the index set is {\em composite}.

Every index set $I$ is associated to a unique edge-labelled rooted tree $T=T(I)$ as follows: The vertex set
of $T(I)$ is the set  $V=V(I)$ 
of all strings that are prefixes of some string in $I$ (including all strings in $I$ and the empty string $\Lambda$).
Say that the parent of the nonempty string $s$ is the prefix of $s$ obtained by removing the final character, and
label the edge by the final character of $s$.  The leaf set of $T(I)$ is the set $I$,  the root is $\Lambda$,
and for every node $v$, $v$  is equal to the sequence of edge labels from the root to $v$.
The labels of the edges leaving any internal node $v$ are distinct.   Under this correspondence
an elementary index set corresponds to a rooted star, and the empty index set corresponds
to the empty indexed tree which consists of a root alone.
\fi

We define an {\em indexed tree} to be a rooted tree $T$ with labelled edges such that for each internal node
the edges  to its children 
have distinct labels.  
 For a node $v$, we write $C(v)=C_T(v)$ for the set of children of $v$ and
$I(v)=I_T(v)$ for the set of labels on the edges from $v$ to $C(v)$.
It follows that, for any node $v$, the sequence of edge labels along the path from the
root to $v$  uniquely identifies $v$, and we identify $v$ we this sequence.  Thus the root
of the tree is the empty sequence $\Lambda$, and for each internal node $v$, the children of $v$
are nodes of the form $vs$ where $s \in I(v)$. We write $L(T)$ for the set of leaves of $T$,
and $Int(T)$ for the set of internal nodes (non-leaves) of $T$.  
The set $L(T)$ is the index set associated with $T$.
In what follows we switch freely between the notion of index set and indexed tree. We also restrict attention to trees of uniform depth, that is where all leaves are at the same distance from the root.

\iffalse
 The above map $I \longrightarrow T(I)$ maps every index set to an indexed tree, and the map is bijective with the inverse map being $T \longrightarrow L(T)$. In what follows we switch freely between an index set  and its associated indexed tree.

For a node $v$ in the indexed tree $T$, $I(v)$ denotes the set of labels on the edges coming out of $v$, which are the symbols $s$ such that $vs$ is a node. We also define $C(v)$ to be the set of nodes which are the children of $v$ in the tree.
\fi

We now define compositions of indexed trees.  
If $T$ is an indexed tree and $(T_v:v \in L(T))$ is a family of trees indexed by $L(T)$, then the composition
 $T_{\circ}=T(T_v:v \in L(T))$ is the indexed tree obtained 
by identifying each leaf $v$ of $T$ with the root of $T_v$.
The index set $L(T_{\circ})$ associated with $T_{\circ}$ 
is the set of all strings of the form $vw$ where $v$ is a leaf of $T$ and $w$ is a leaf of $T_v$.

Every tree $T$ can be constructed as a composition of the star from the root,
with the collection of subtrees rooted at the children of the root.  By applying this decomposition recursively,
we can build up every tree   from the collection of stars corresponding
to each internal vertex.

In the special case that all $T_v$ are the same tree $T'$, we say the composition
is {\em uniform} and write it as $T \composedwith T'$.  This composition
clearly forms an associative operation on indexed trees so that the notation
$T_1 \composedwith \cdots \composedwith T_k$ is well defined.  The leaf set
of $T_1 \circ \cdots \circ T_k$ consists of sequences $v_1,\ldots,v_k$
where $v_i$ is a leaf of $T_i$. Such trees may be thought of as representations of product sets $I_1 \times I_2 \times \cdots \times I_k$.
If all $T_i$ are the same tree $T$, we write this composition as $T^{(k)}$, which is the
{\em $k$-wise iterated composition of $T$}.

\subsubsection{Compositions of various objects}

With the framework of indexed trees, we now define notions
of compositions for various types of objects over index sets.
For an appropriate object type $\tau$ the form of the composition
is the same.  Every object of type $\tau$ is defined with respect to an index set,
and the index set is represented as the leaf set of a tree.  For simplicity
we say that object $\omega$ is defined over $T$ if its index set is $L(T)$.

Let $T$ be a tree, and let $(T_v:v \in L(T))$
be a family of trees indexed by the leaves of $T$.  We defined the composition of $T$ with $(T_v:v \in L(T))$
to be the tree $T_{\circ}$ obtained by identifying the roots of each $T_v$ with the leaf $v$ of $T$. Recall that $L(T_{\circ})$
consists of pairs $vw$ where $v \in L(T)$ and $w \in L(T_v)$.

Let $\tau$ be some type of object (such as hypergraph) over an index set.  
Suppose
that $\omega$ is an object of type $\tau$ over the index set $L(T)$ and for each $v \in L(T)$ let  $\omega_v$
be an object of type $\tau$ over the index set $L(T_v)$.     For certain types
$\tau$ we define a composition $\omega_{\circ}=\omega(\omega_v:v \in L(T))$ over $L(T_{\circ})$.
 Our composition operation for $\tau$ will
combine these objects into an object $\omega_{\circ}$ over index set $L(T_{\circ})$.

Here are compositions for some basic object types:

\begin{itemize}
\item {\em Weight functions.}  If $\tau$ is the class of weight functions, $w$ is a weight function on $T$,
and for each $v \in L(T)$, $w_v$ is a weight function on $L(T_v)$, then the composition 
$w_{\circ}=w(w_v:v \in L(T))$  is the weight function on $L(T_{\circ})$ 
where $\omega_{\circ}(vw):=\omega(v)\omega_v(w)$.
\item {\em Subsets.}
By associating a subset of a set $I$ with the weight function given by its characteristic function, the composition of weight functions gives a notion of composition of subsets. 
\item {\em Weight function families}.  Let $\Omega$ be a family of weight functions on $L(T)$
and for each $v \in L(T)$ let $\Omega_v$ be a family of weight functions over $L(T_v)$.
Then $\Omega_{\circ}=\Omega(\Omega_v:v \in L(T))$ is the weight function family
on $L(T_{\circ})$ consisting of all compositions $w(w_v:v \in L(T))$ where $w \in \Omega$
and for each $v \in L(T)$, $w_v \in \Omega_v$.
\item {\em Hypergraphs}.
  By viewing a hypergraph as a set of boolean weight functions, the notion
of composition of weight function families specializes to a notion of composition of hypergraphs.
\item {\em Boolean functions.} Let $f$ be a boolean function over $L(T)$ and, for each $v \in L(T)$, let
$f_v$ be a boolean function over $L(T_v)$.  Then the composition $f_{\circ}=f(f_v:v \in L(T))$
is the boolean function defined over $L(T_{\circ})$ whose value on a boolean assignment over $L(T_{\circ})$ is
computed by defining $b_v$ for $v \in L(T)$ to be $f_v$ evaluated on the subset of inputs corresponding to $L(T_v)$
and then evaluating $f$ on assignment $(b_v:v \in L(T))$.
\iffalse
\item {\em Assignment selectors.} There seems to be no notion of compositions of assignments that is useful for our purposes.
However, there is a very useful notion of composition of assignment selectors.  Let $\vec{\alpha}$ be an assignment
selector over $L(T)$ and for each $v \in L(T)$, let $\vec{\alpha}_v$ be an assignment selector over $L(T_v)$.
The composition $\vec{\alpha}_{\circ}$ is an assignment selector over $L(T)_{\circ}$.  The assignments  $\alpha_{\circ}^b$
for $b \in \{0,1\}$ are obtained as follows.  Label the root of $T$ by the bit $b$, and then label the children of the root
according to the assignment $\alpha^b$.  Then for each child $v$ of the root, label $L(T_v)$ according to 
$\alpha_v^c$ where $c=\alpha^b_v$.
\fi
\end{itemize}

The notions of uniform and iterated compositions are defined in the natural way.
If $T=T_1 \composedwith \cdots \composedwith T_k$ is a uniform composition of trees
and for each $i \in [k]$, $\Omega_i$ is an object of type $\tau$ over $L(T_i)$
then $\Omega_1 \composedwith \cdots \composedwith \Omega_k$ is an object of type $\tau$
over $T$.  It is easy to verify that for the various compositions we define 
that the operation $\composedwith$ is associative so that the uniform composition is well-defined.
We write $\Omega^{(k)}$ for the $k$-wise iterated composition of $\Omega$.

Given an arbitrary indexed tree $T$,
a {\em $T$-ensemble} of objects of type $\tau$ is an indexed family $\omega_T=(\omega_v:v \in Int(T))$ where $\omega_v$ is an object of type $\tau$ over the index set $I(v)$.   We define the composition of $\omega_T$, denoted $\composition \omega_T$
inductively:  For a null tree (consisting of only a root $\Lambda$ so that $L(T)=\{\Lambda\}$), there is a null object of type $\tau$.  For weight functions, the null object is the weight function mapping $\Lambda$ to 1, and for boolean functions the null
object is the (univariate) identity function.  For a non-null tree $T$, $\composition \Omega_T$ is given by
$\Omega_{\Lambda}(\Omega_v:v \in C(\Lambda))$, which is the composition of the object associated with the root
with the collection of objects associated with the children of the root.  Unwinding this recursion gives the
following alternative description of the compositions of $T$-ensembles for various objects:

\begin{itemize}
\item {\em Weight functions.}  Let $w_T=(w_v:w \in Int(T))$ be a $T$-ensemble of weight functions (so that $w_v$ is a weight function on $I(v)$).  We can view $w_v$ as assigning a weight to each edge coming out of $v$. 
Then the composition $\composition w_T$
assigns a weight to each leaf $l$ which is given by the product of the weights on the edges along
the path from the root to $l$. 
\item {\em Subsets.}
By associating a subset of a set $I$ with the weight function given by its characteristic function, the composition of weight functions gives a notion of composition of subsets. 
\item {\em Weight function families}.  Let $\Omega_T=(\Omega_v:v \in Int(T))$ be a $T$-ensemble of weight
function families.  Thus, for each $v$, $\Omega_v$ is a set of weight functions  over  $I(v)$. The composition 
$\composition \Omega_T$ is the set of all weight functions of the form $\composition w_T$ where $w_T$ is a weight function ensemble
satisfying $w_v \in \Omega_v$ for each $v$.
\item {\em Hypergraphs}.
  Let $\mc{H}_T=(\mc{H}_v:v \in Int(T)\}$ be a $T$-ensemble of hypergraphs.  Thus
for each $v$, $\mc{H}_v$ is a hypergraph on $I(v)$, which can be viewed as a family of boolean-valued
weight functions on $I(v)$.  The composition 
$\composition \mc{H}_T$  is obtained by specializing the composition of weight function families, and is a  hypergraph on 
$L(T)$
\item {\em Boolean functions.} Let $f_T=(f_v:v \in Int(T))$ be a $T$-ensemble of
boolean functions.  The composition  $\composition f_T$  is the function over $L(T)$ obtained
by viewing $T$ as a circuit and each vertex $v$ as a gate which computes the function $f_v$.
\item {\em Assignment selectors.}  This is described more easily in the context of the next subsection, so we present it there.
\end{itemize}

\subsection{Boolean labelings of trees, and compositions of assignment selectors}
\label{Boolean labelings}

A {\em boolean $T$-labeling} for an indexed tree $T$ is a mapping $b_T=(b(v):v \in T)$
that assigns a bit to each vertex of $T$.  Given a boolean $T$-labeling $b_T$ we define $b_v$, for an internal
node $v$, to be the  labeling $b_T$ restricted to the children of $v$.  Note the difference
between the notation $b(v)$, which is a single bit, and $b_v$, which is an assignment to $C(v)$.
Thus, any boolean $T$-labeling $b_T$ 
induces an {\em assignment $T$-ensemble} $(b_v : v \in Int(T))$.   Also, the leaf assignment
determined by $b_T$ is the boolean assignment to the leaves obtained by restricting $b_T$ to $L(T)$.

Boolean $T$-labelings will arise for us in two ways: 
\begin{itemize}
\item (Bottom-up labelings) 
If $f_T =(f_v:v \in Int(T))$ is a $T$-ensemble of boolean functions and $\alpha \in \{0,1\}^{L(T)}$
is a boolean assignment to the leaves of $T$ then viewing $f_T$ as a circuit with node $v$
being a gate computing $f_v$, then the evaluation of $f_T$ on input $\alpha$, denoted $f_T(\alpha)$,
is a boolean $T$-labeling $(b(v):v \in T))$ where the label $b(v)$ for a node $v$ is defined from the leaves upward
as follows: if $v$ is a leaf then $b(v)=\alpha(v)$ and if $v \in Int(T)$, then having defined $b(w)$ for each child $w$ of $v$
we have $b(v)=f_v(b_v)$, where as above $b_v$ is the assignment to $C(v)$ determined by $b$.
We call the resulting boolean $T$-labeling
the {\em evaluation  labeling induced by $f_T$ and $\alpha$.}
  
\item (Top-down labelings)  
If $c \in \{0,1\}$ and  $\vec{\alpha}_T=(\vec{\alpha}_v:v \in Int(T))$ is a $T$-ensemble of
assignment selectors, then $c$ and $\vec{\alpha}_T$ induce a boolean $T$-labeling in the following way.
Label the root by $b(r)=c$.  Now starting from the root apply the following procedure: Having labeled an internal
node $v$ by $b(v)$, use the bit $b(v)$ to select the assignment $\alpha_v^{b(v)}$ from the assignment
selector $\vec{\alpha}_v$, and then label the children of $v$ according to $\alpha_v^{b(v)}$.
We call this the boolean $T$-labeling induced by root label $c$ and $\vec{\alpha}_T$.
Note that by definition the assignment $T$-ensemble associated to $b$, $(b_v:v \in Int(T))$ is
given by $b_v=\alpha_v^{b(v)}$.

If we don't specify a bit $c$, then the $T$-ensemble of assignment selectors $\vec{\alpha}_T$
defines  a boolean $T$-labeling selector $\vec{b}_T=(b^0_T,b^1_T)$,
where for $c \in \{0,1\}$, $b^c_T$ is the boolean $T$-labeling induced by $c$ and $\vec{\alpha}_T$.
We call this the boolean $T$-labeling selector induced by $\vec{\alpha}_T$.
\end{itemize}

The top-down construction  implicitly provides a natural notion of composition
of assignment selectors:

{\em Composition of assignment selectors.}
Let $\vec{\alpha}_T=(\vec{\alpha}_v:v \in Int(T))$ be a $T$-ensemble of assignment selectors.
Let $\vec{b}_T$ be the boolean $T$-labeling selector induced by $\vec{\alpha}_T$.
If we restrict each of the labelings $\vec{b}^1_T$ and $\vec{b}^0_T$ to $L(T)$
we get an assignment selector over
 $L(T)$.  This assignment selector is defined to be the composition of the ensemble
$\vec{\alpha}_T$ and is denoted $\vec{\alpha}_{\circ}=(\alpha^0_{\circ},\alpha^1_{\circ})=\composition \vec{\alpha}_T$.

As with compositions of other objects, we specialize composition of assignment selectors to the case
of uniform compositions, and denote a uniform composition of assignment
selectors by $\vec{\alpha}_1 \composedwith \cdots \composedwith \vec{\alpha}_k$.

Observe that the bottom-up labelings and top-down labelings fit together in the following way.
Suppose $f_T$ is a $T$-ensemble of boolean functions and $\vec{\alpha}_T$ is a $T$-ensemble of assignment selectors. 
 Suppose that
for each vertex $v \in Int(T)$, $\vec{\alpha}_v$ is $f_v$-compatible, which we defined earlier to mean
$f_v(\alpha^0_v)=0$ and $f_v(\alpha^1_v)=1$.    In this case we say
that the ensemble $\vec{\alpha}_T$ is $f_T$-compatible. 

\begin{proposition}
\label{T-labeling}
Let $f_T$ be a $T$-ensemble of boolean functions and $\vec{\alpha}_T$ be a $T$-ensemble of assignment selectors.
Let $\vec{b}_T$ be the boolean $T$-labeling selector induced (top-down) by $\vec{\alpha}_T$ and let $\vec{\alpha}_{\circ}$
be the composition  $\composition \vec{\alpha}_T$ (which was defined to be the restriction of  $\vec{b}_T$ to $L(T)$). 
If for each $v \in Int(T)$, $\vec{\alpha}_v$ is $f_v$ compatible with $f_v$ then:

\begin{itemize}
\item For $c \in \{0,1\}$, The (bottom-up) labeling induced by $f_T$ and $\alpha_{\circ}^c$ is $b_T^c$.
\item The composed assignment selector $\composition \vec{\alpha}_T$ is  $F$-compatible, where $F=\composition f_T$ 
is the composition of $f_T$.
\end{itemize} 
\end{proposition}

\begin{proof}
For $c \in \{0,1\}$, let $a_T^c$ be the labeling induced by $f_T$ and 
$\alpha_{\circ}^c$. We prove that, for all $v \in T$, $a_T^c(v)=b_T^c(v)$. We proceed by induction
on the size of the subtree rooted at $v$.     
For $v \in L(T)$ we have $a_T^c(v)=\alpha^c_{\circ}(v)=b_T^c(v)$.  For $v \in Int(T)$, by the induction hypothesis,
we have $a_T^c(w)=b_T^c(w)$ for all children $w$ of $v$, equivalently, we have $a_v=b_v$.  Now by the definition
of $a_T$ we have $a(v)=f_v(a_v)=f_v(b_v)$.  On the other hand, by the definition of $b_T$
we have that $b_v$ is equal to $\alpha_v^{b(v)}$, and since $\vec{\alpha}_v$ is compatible with $f_v$
this implies $f(b_v)=b(v)$ and so $b(v)=a(v)$, as required.

For the second part, the composed assignment $\vec{\alpha}_{\circ}$  is (by definition)  equal to $\vec{b}_T$ restricted to
$L(T)$, and by the first part this is the same as $\vec{a}_T$ restricted to $L(T)$.  For each $c \in \{0,1\}$,
the value of $f_T$ at $a_T^c$ is the value of the root in the bottom-up labeling (viewing $T$ as a circuit
with gates $(f_v:v \in Int(T))$), and this is the label given to the root by $a_T^c=b_T^c$
which is equal to $c$ by the definition of the top-down labeling $b_T^c$

\end{proof}

\iffalse
For each of these compositions,
we have notions of uniform and iterated composition which parallel
the same notions for 
index sets and indexed trees.   As before, let $\tau$ be some type of object defined over
index sets, let $T$ be a tree and for each $v \in Int(T)$ let $\beta_v$ be an object
of type $\tau$ over $I(v)$. 
In the case that the tree has depth 2, the composed object
is denoted by  $\alpha(\beta_i:i \in I)$ where $I$ is the index set $I(r)$ at the root,
$\alpha=\beta_r$ is the object associated with the root, and for each $i \in I$, $\beta_i$
is the object associated with $I(i)$.    In the case that $T$ is a uniformly composed tree of
depth $k$ and all of the nodes at depth $d$ are associated with the same object $\beta_d$
we write $\beta_1 \composedwith \cdots \composedwith \beta_k$ for the associated composition.
If further, $T$ is the $k$-wise iterated composition of a tree $U$, and all of the $\beta_j$
are the same object $\beta$, we write $\beta^{(k)}$ for the composition, and call
it the {\em $k$-wise iterated composition of $\beta$}.  
\fi

\section{The growth of various complexity measures under iterated composition} \label{sec:Main}

Our goal in this section is to understand how $m(f^{(k)})$ relates to $m(f)$ for various complexity measures. In Section \ref{subsec:analysing} we provide a high level discussion of how to analyse $m(f^{(k)})$. To do so we will initially state, without proofs, the lemmas which lead to the main result. We will then prove the main theorem modulo these lemmas. Finally, in Section \ref{sec:details}, we will provide all the remaining proofs and definitions which were left out.

\subsection{Analysing $m(f^k)$} \label{subsec:analysing}
Fix an assemblage $\mc{A}$ and let $m$ be the associated complexity measure.  Informally, $m(f)$ is high
if there is a hard input $\alpha$, which means that
 every weight function in $\mc{A}_{\alpha}(f)$ has high total weight.

Let's start with the most general  function composition, where $F$ is the composition  of a $T$-ensemble 
$f_T = (f_v: v\in Int(T)\})$  for an arbitrary tree $T$ of uniform depth.   Fix an input $\alpha$ to the leaves of $T$ and let $(b(v) : v \in T)$ be the evaluation labeling of $f_T$ on $\alpha$ and $(b_v : v \in Int(T))$ be the corresponding assignment ensemble. Recall that, for each $v \in Int(T)$, we have $b(v)=f_v(b_v)$.

To determine $m_{\alpha}(F)$ we want to determine the minimum weight of a weight function
$w \in \mc{A}_{\alpha}(F)$.  In trying to analyze this minimum, it is natural to look at
weight functions which are representable as $T$-compositions 
of weight functions as follows: For each internal node $v$ of $T$
select a weight function $w_v$ that belongs to the set of weight functions $\mc{A}_{b_v}(f_v)$, and
take $w$ to be the composition of the  $T$-ensemble  $(w_v:v \in Int(T))$.  
We say that
such an ensemble is {\em compatible with $(f_v:v \in Int(T))$ and $\alpha$}.
Our hope is that the minimum weight of a weight function in $\mc{A}_{\alpha}(F)$
is attained by such a composition.
 It is easy to show that this is  true provided that the assemblage $\mc{A}$
satisfies the following two properties: 
\begin{enumerate}
\item For any weight function ensemble that is compatible
with $(f_v:v \in Int(T))$ and an assignment $\alpha$, its composition belongs $\mc{A}_{\alpha}(F)$.
\item If $w$ is any weight function in $\mc{A}_{\alpha}(F)$ then there is a weight function 
ensemble $(w_v:v \in Int(T))$ that is compatible with $(f_v:v \in Int(T))$ and $\alpha$
whose composition has total weight  less than that of $w$.
\end{enumerate}

We call an assemblage {\em well behaved} if it satisfies (1) and (2).
Summarizing the above, we have:

\begin{proposition}
\label{prop:well-behaved}
Let $m$ be a complexity measure associated to a well-behaved assemblage.  Let $(f_v:v \in Int(T))$ be a $T$-ensemble of functions
with composition $F$ and let $\alpha$ be an input to $F$.  Then

\[
m_{\alpha}(F) = \min |w|,
\]
where $w$ ranges over all compositions of weight function ensembles $(w_v:v \in Int(T))$ that are
compatible with $(f_v:v \in Int(T))$ and $\alpha$.
\end{proposition}  

In section \ref{sec:details}, we will prove

\begin{restatable}{lemma}{assemblagecomp}
%\begin{lemma}
\label{assemblage comp}
Each of the assemblages $\partial\mc{B}$, $\wit{}$, $\witstar{}$ and $\Psi$ are well-behaved.
%\end{lemma}
\end{restatable}
This implies that Proposition \ref{prop:well-behaved} can be applied to certificate complexity,
fractional certificate complexity and sensitivity.  In the remaining discussion we assume
that the assemblage $\mc{A}$  is well-behaved.

At this point, we restrict attention to 
 $F$ which are uniform compositions
$F=f_1 \composedwith \cdots \composedwith f_k$ where $f_i$ is a boolean function
over the index set $L(T_i)$ where $T_i$ is an indexed star.  To understand $m(F)$
we want to identify an input $\alpha$ that maximizes $m_{\alpha}(F)$.  Actually we'll try
to identify an assignment selector $\vec{\alpha}=(\alpha^0, \alpha^1)$ for $F$ that is $m$-optimal for $F$,
which (as defined in Section \ref{sec:complexmeasures}) means that
$\alpha^0$ maximizes $m_{\alpha}(F)$
over $\alpha \in F^{-1}(0)$ and $\alpha^1$ maximizes $m_{\alpha}(F)$ over $\alpha \in F^{-1}(1)$.
It is natural to speculate that we can obtain such an assignment selector $\vec{\alpha}$
as a composition of assignment selectors
$\vec{\alpha}_1 \composedwith \cdots \composedwith \vec{\alpha}_k$ where $\vec{\alpha}_i$ is an assignment
selector for $f_i$.  This indeed turns out to be the case, as is stated in the second part of the following lemma:
\iffalse
is an input to $f_i$.  This does not work; in particular, it is not clear what notion of composition
should be used. It turns out that something like this works if we consider
assignment selectors rather than assignments. In Section \ref{sec:details}, we define a composition operation $\circ$ on assignment selectors. For the purposes of this discussion it is enough to note that this operation satisfies the following properties:
\fi

\begin{restatable}{lemma}{uniformcomp}
\label{lem:uniform comp}
Let $f_1,\ldots,f_k$ be a sequence of boolean functions and let $F$ be their composition.  

\begin{itemize}
\item  
If $\vec{\alpha}_1,\ldots,\vec{\alpha}_k$ are assignment selectors, where $\vec{\alpha}_i$ is $f_i$-compatible, then $\vec{\alpha} := \vec{\alpha}_1 \composedwith \cdots \composedwith \vec{\alpha}_k$ is an $F$-compatible assignment selector.
\item 
There are assignment selectors $\vec{\alpha}_1,\ldots,\vec{\alpha}_k$, where $\vec{\alpha}_i$ is $f_i$-compatible, such that $\vec{\alpha}:= \vec{\alpha_1} \composedwith \cdots \composedwith \vec{\alpha_k}$ is an $m$-optimal selector for $F$.
\end{itemize}
%\end{lemma}
\end{restatable}

Note that it is not the case in this lemma that each $\vec{\alpha}_i$ in the conclusion is $m$-optimal for $f_i$.  Nevertheless, the lemma is useful because,
in evaluating $m(F)$, it is enough to consider all assignment selectors 
of the form $\vec{\alpha}_1 \composedwith \cdots \composedwith \vec{\alpha}_k$ where
each $\alpha_i$ is $f_i$-compatible.

Now consider such
a composed assignment selector $\vec{\alpha}$.  We want to understand $m_{\alpha^0}(F)$
and $m_{\alpha^1}(F)$ and for this it suffices to  determine the weight functions
$ w^0 \in \mc{A}_{\alpha^0}(F)$ and $w^1 \in \mc{A}_{\alpha^1}(F)$  of minimum weight.
One might hope that $w^0$ and $w^1$ can each be expressed as a uniform composition
of weight functions, but this need not be true.  Once again we need to consider
compositions of weight function selectors rather than weight functions.
There is  a natural way to compose any sequence
of weight function selectors, but it is a bit complicated because it depends not only on
the sequence $\vec{w}_1,\ldots,\vec{w}_k$ but also on the associated assignment
selectors $\vec{\alpha}_1,\ldots,\vec{\alpha}_k$.   What we end up with is a composition
operation which acts on pairs  $(\vec{\alpha}_i,\vec{w}_i)$ consisting of an assignment selector and weight function
selector for $f_i$  and produces such a pair
$(\vec{\alpha},\vec{w})$ for $f$. We call this an assignment-weight selector-pair, or simply {\em AW-selector pair}.
The assignment selector $\vec{\alpha}$ is just the composition of the assignment
selectors $\vec{\alpha}_i$ as before (and does not depend on the $\vec{w}_i$, but $\vec{w}$
depends on both the $\vec{\alpha}_i$ and $\vec{w}_i$). Again, due to the technical nature of this construction, we delay the explicit definition for Section \ref{sec:details}. For now it is enough to note that this composition operation satisfies the following properties (see Section \ref{sec:assembmeasures} for the definition of $(f,\mc{A})$-compatible):

%We say that an AW-selector pair  is {\em compatible with the function $f$ and assemblage $\mc{A}$}, or simply {\em $(f,\mc{A})$-compatible},
%provided that $\vec{\alpha}$ is $f$-compatible and for each $b \in \{0,1\}$, $w^b \in \mc{A}_{\alpha^b}(f)$.

\begin{restatable}{lemma}{hittingdecomp}
\label{lem:hitting decomp}
Let $f_1,\ldots,f_k$ be a sequence of boolean functions and let $F$ be their composition.  Let $\vec{\alpha}_1,\ldots,\vec{\alpha}_k$ be assignment selectors
such that each $\vec{\alpha}_i$ is $f_i$-compatible. Then the following hold:

\begin{itemize}
\item  For any sequence $\vec{w}_1,\ldots,\vec{w}_k$ of weight function selectors, where $(\vec{\alpha}_i,\vec{w}_i)$ is $(f_i,\mc{A})$ compatible,
the composition $(\vec{\alpha},\vec{w}):= (\vec{\alpha}_1,\vec{w}_1) \composedwith \cdots \composedwith (\vec{\alpha}_k,\vec{w}_k)$ is $(F,\mc{A})$-compatible.
\item
There are weight function selectors $\vec{w}_1, \cdots, \vec{w}_k$ such that the weight function $\vec{w}$ that comes from the composition
 $(\vec{\alpha},\vec{w}):=(\vec{\alpha_1},\vec{w}_1) \circ \cdots \circ (\vec{\alpha}_k,\vec{w}_k)$  is an $m$-optimal weight function selector for $F$ at $\vec{\alpha}$.
\end{itemize}
\end{restatable}
%\end{lemma}

We now define the following function, which maps a sequence of AW selector-pairs to a real number:

\[
V\left((\vec{\alpha}_1,\vec{w}_1),\ldots,(\vec{\alpha}_k,\vec{w}_k)\right) := \max\{|w^0|,|w^1|\},
\]
where the weight function selector $\vec{w}$ is given by $(\vec{\alpha},\vec{w})=(\vec{\alpha}_1,\vec{w}_1) \composedwith \cdots \composedwith (\vec{\alpha}_k,\vec{w}_k)$. Combining Lemmas \ref{lem:uniform comp} and \ref{lem:hitting decomp} we obtain

\begin{restatable}{lemma}{mclassified} \label{lem:m classified}
Let $F=f_1 \composedwith \cdots \composedwith f_k$.  Then $m(F)$ is equal to the maximum, over all sequences of assignment selectors
$\vec{\alpha}_1,\ldots,\vec{\alpha}_k$ where each $\vec{\alpha}_i$ is $f_i$-compatible, of the minimum, over all sequences $\vec{w}_1,\ldots,\vec{w}_k$
of weight function selectors such that $(\vec{\alpha}_i,\vec{w}_i)$ is $(f_i,\mc{A})$-compatible, of $V\left((\vec{\alpha}_1,\vec{w}_1),\ldots,(\vec{\alpha}_k,\vec{w}_k)\right)$.
%\end{lemma}
\end{restatable}

So next we want to understand the function $V$.  The following definitions will be helpful.
\begin{itemize}
\item {\em Largest eigenvalue.} For a square real matrix $A$, $\rho(A)$ denotes the
maximum of $|\lambda|$ over all eigenvalues $\lambda$ of $A$.
\item
The {\em profile matrix} of an AW selector-pair $(\vec{\alpha},\vec{w})$ is defined to be the 2 by 2 matrix $M_{\vec{\alpha},\vec{w}}$ with rows and columns indexed by $\{0,1\}$ with
$s,t$ entry equal to $\sum_j  w^s(j)$ where the sum ranges over indices $j$ such that $\alpha^s_j=t$.  Note that for $s \in \{0,1\}$, the $s$th row sum of $M_{\vec{\alpha},\vec{w}}$
is equal to $|w^s|$.
\item {\em Profile matrix family  $\mc{M}_{\vec{\alpha}}(f)$ for the assemblage $\mc{A}$.} 
For each function $f$ and assignment selector $\vec{\alpha}$, 
$ \mc{M}_{\vec{\alpha}}(f)$ is the set of all
matrices $M_{\vec{\alpha},\vec{w}}$ where $\vec{w}$ ranges over weight function selectors
such that $(\vec{\alpha},\vec{w})$ form an $(f, \mc{A})$-compatible pair. Note that the set of profile matrices depends on the assemblage. In particular, the set of profile matrices for the assemblage $\mc{W}$ will be a subset of the set of profile matrices for the assemblage $\mc{w}^*$.
  
\end{itemize}
\noindent
{\bf Example.} Consider $f = \text{NAND}_n(x)$, that is $f(x) = 0$ only at the all 1's input. Let $\alpha^0 = (1,1,\cdots, 1)$ and $\alpha^1 = (0,1,1,1,\cdots,1)$. Take $m$ to be certificate complexity. The only witness for $\alpha^0$ is $w_0 \equiv 1$, and we take $w_1$ to assign weight $1$ on the 0 index and weight 0 otherwise. Here the profile matrix
\[M_{\vec{\alpha},\vec{w}} = 
  \begin{bmatrix}  
     0 & n \\
     1 & 0 \\
  \end{bmatrix}.
\]
Note that $\rho(M_{\vec{\alpha},\vec{w}}) = \sqrt{n}$. It is also true that $C^{\lim}(f) = \sqrt{n}$ (this is not a coincidence as we will see).

It turns out that matrix multiplication captures the mechanism behind the composition of AW selector-pairs. In fact, for any sequence $\{(\vec{\alpha}_i,\vec{w}_i)\}_{i=1}^k$ we have that $V((\vec{\alpha}_1,\vec{w}_1),\ldots,(\vec{\alpha}_k,\vec{w}_k))$ is equal to the maximum row sum of the product $M_{\vec{\alpha}_1,\vec{w}_1} \cdots M_{\vec{\alpha}_k,\vec{w}_k}$. This follows from

\begin{restatable}{proposition}{matrixmult} \label{prop:matrix mult}
For any sequence $(\vec{\alpha}_1,\vec{w}_1),\ldots,(\vec{\alpha}_k,\vec{w}_k)$ of AW selector-pairs, if $(\vec{\alpha},\vec{w})$ is their composition
then the profile matrix $M_{\vec{\alpha},\vec{w}}$ is given by:

\[
M_{\vec{\alpha},\vec{w}}=M_{\vec{\alpha}_1,\vec{w}_1} \cdots M_{\vec{\alpha}_k,\vec{w}_k}.
\]
%\end{proposition}
\end{restatable}

Lemma \ref{lem:m classified} and Proposition \ref{prop:matrix mult} imply
\begin{corollary}
\label{cor:m(f) via matrices}
Let $F=f_1 \composedwith \cdots \composedwith f_k$.  Then $m(F)$ is equal to the maximum, over all sequences $\vec{\alpha}_1,\ldots,\vec{\alpha}_k$ where $\vec{\alpha}_i$ is $f_i$-compatible, of the minimum, over all choices of matrices $M_1,\ldots,M_k$ where $M_i$ belongs to the profile matrix family $\mc{M}_{\vec{\alpha}}(f_i)$,
of the maximum row sum of the product $M_1 \cdots M_k$.
\end{corollary}

At last we are ready to consider the case of iterated composition, where all of the $f_i$ are the same function $f$. We wish to understand $m(f^{(k)})$, which we now know is the maximum over choices of $\vec{\alpha}_i$ of the minimum over choices of $\vec{w}_i \in \mc{A}_{\vec{\alpha}_i}(f)$ of $V((\vec{\alpha}_1,\vec{w}_1) \circ \cdots \circ (\vec{\alpha}_k,\vec{w}_k))$. Intuitively, one may think of the choice of each assignment selector $\vec{\alpha}_i$ as defining the set of profile matrices $\mc{M}_{\vec{\alpha}_i}(f)$, and the choice of each weight function selector $\vec{w}_i$ as choosing a matrix $M_i \in  \mc{M}_{\vec{\alpha}_i}(f)$.

If we wish to lower bound $m(f^{(k)})$ we hope to find assignment selectors $\vec{\alpha}_i$ for which all possible products of the form 
\[ M_1M_2\cdots M_k, \ \ \ \ \ \ M_i \in \mc{M}_{\vec{\alpha}_i}(f)\]
have a large max row sum. We will need the following simple fact about matrices:

\begin{fact} \label{fact:matrix rho}
  For any matrix $M \in \mathbb{R}^{2 \times 2}_{\geq 0}$ we have
 
   \[||M||_{\infty} \geq \rho(M)/2.\]

\end{fact}
It turns out that the minimum of $\rho(M)$ over matrices in the family $\mc{M}_{\vec{\alpha}}(f)$ gives a good notion of how hard the input selector $\vec{\alpha}$ is.
  We now introduce two more definitions:

\begin{itemize}
\item {\em The characteristic value $\hat{m}_{\vec{\alpha}}(f)$ for $(f,\vec{\alpha})$}.
For a boolean function $f$ and assignment selector $\vec{\alpha}$ compatible with $f$, define $\hat{m}_{\vec{\alpha}}(f)$ to
be the minimum of $\rho(A)$ over all $A \in \mc{M}_{\vec{\alpha}}(f)$. 
The function $\hat{m}_{\vec{\alpha}}(f)$ can be viewed as another complexity measure derived from
the assemblage $\mc{A}$ which is {\em bi-local} rather than {\em local} in the sense
that it depends on a pair of assignments rather than just one. 
 \item {\em The characteristic value $\hat{m}(f)$} is the maximum, over all assignment selectors $\vec{\alpha}$ compatible with
$f$, of $\hat{m}_{\vec{\alpha}}(f)$.    
\end{itemize}

 If we can prove that $m_{\vec{\alpha}}(f) \geq \lambda$ then this will imply by Fact \ref{fact:matrix rho} that $\max(m_{\alpha_0}(f), m_{\alpha_1}(f)) \geq \lambda/2$. This suggests that, in order to construct a hard assignment selector for $f^{(k)}$, we should choose a selector $\vec{\beta}$ which maximizes $\hat{m}_{\vec{\beta}}(f)$ and then set $\vec{\alpha} = \vec{\beta}^{(k)}$. We will use this idea to prove
\begin{restatable}{lemma}{lowerbound} \label{lem:m lower bound}
  For any boolean function $f$ and $k \in \mathbb{N}$
   \[m(f^{(k)}) \geq \hat{m}(f)^k/2.\]
\end{restatable}

%The proof follows by showing that if $\vec{\alpha} = \vec{\beta}^k$, where $\vec{\beta}$ was defined above, then $\hat{m}_{\vec{\alpha}}(f^k) \geq %\lambda^k$. This reduces to showing that for any matrix of the form 
 % \[ M = M_1M_2 \cdots M_k, \ \ \ \ \ \ M_i \in \mc{M}_{\vec{\beta}}(f) \]
 %has $\rho(M) \geq \lambda^k$. To prove this, one would like to use the property that $\rho(M_i) \geq \lambda$ for each $i$ (this follows by definition of %$\vec{\beta}$). However, in general this is not enough to guarantee that $\rho(M) \geq \lambda^k$. Luckily, due to the fact that the $M_i$ are profile matrices for %the selector $\vec{\beta}$, there is additional structure available to imply the desired result. Again the detailed proof is delayed for Section \ref{}.
 
 Next we obtain an upper bound $m(f^{(k)})$. We show that for any sequence of assignment selectors $\vec{\alpha}_1, \cdots, \vec{\alpha}_k$, we can find matrices $M_i \in \mc{M}_{\alpha_i}(f)$ such that the product $M = M_1 M_2 \cdots M_k$ has no entry larger than $nk (\hat{m}(f))^{k-1}$. This will prove

 \begin{restatable}{lemma}{upperbound} \label{lem:m upper bound}
   For any boolean function $f$ on $n$ variables and $k \in \mathbb{N}$,
     \[m(f^{(k)}) \leq 2n k(\hat{m}(f))^{k-1}.\]
 %\end{lemma}
\end{restatable} 
%To sketch the proof or this, first recall by Lemma \ref{lem:uniform comp} that there is an $m$-optimal selector for $f^k$ of the form $\vec{\alpha}= %\vec{\alpha}_1 \circ \cdots \circ \vec{\alpha}_k$. Our goal is then to find matrices $M_i \in \mc{M}_{\vec{\alpha}_i}(f)$ such that the max row sum of $M_1 M_2 %\cdots M_k$ is at most $n k\hat{m}(f)^{k-1}$. We know that there is $M_i \in \mc{M}_{\vec{\alpha}_i}(f)$ such that $\rho(M_i) \leq \hat{m}(f)$. However, once %again this fact alone is not enough to imply our general result. We will again have to look at the additional structure of these matrix families that follows from the %fact that they are profile matrices.
 
Armed with these lemmas, we easily obtain the main result of this section: 
 
 \begin{restatable}{theorem}{mlim}
\label{thm:m^lim}
Let $m$ be a complexity measure with associated well-behaved assemblage $\mc{A}$.  Then, for any boolean function $f$, $m^{\lim}(f)=\hat{m}(f)$.
%\end{theorem}
\end{restatable}
\begin{proof}
  Assume $f$ is a function on $n$ variables. Recall that $m^{\lim}(f) := \lim_{k \rightarrow \infty}m(f^{(k)})^{1/k}$. By Lemmas \ref{lem:m lower bound} and $\ref{lem:m upper bound}$ we have that
  \[ \lim_{k \rightarrow \infty}(1/2)^{1/k}(\hat{m}(f)) \leq m^{\lim}(f) \leq \lim_{k \rightarrow \infty}(2n k)^{1/k} \hat{m}(f). \]
  Both the above limits approach $\hat{m}(f)$, thus the result follows.
\end{proof}

This concludes the informal discussion of the main result. We note that Theorem \ref{thm:m^lim} has been proven modulo all lemmas and propositions stated in this section. It remains to explicitly define the composition of assignment selectors, and AW selector-pairs, and provide the proofs which were left out of this section. All such proofs and definitions are in the following section.

\subsection{Filling in the details} \label{sec:details}
   In this section we prove the technical lemmas which were referred to in the previous section.

Our first step is to  prove Lemma \ref{assemblage comp} that sensitivity, certificate complexity, and fractional certificate complexity all induced by well behaved assemblages. The following proposition shows that certificates (respectively fractional certificates) compose and decompose nicely.
 
 \newtheorem*{myprop}{Proposition A}
   \begin{myprop}
\label{prop:wt comp}
Let $T$ be an indexed tree, and let $(\Omega_v:v \in Int(T))$ be an ensemble of boolean valued weight function families (i.e., hypergraphs).
Let $\Omega_T$ be the composition $\composition_T(\Omega_v:v \in Int(T))$.  
\begin{itemize}
\item If $(h_v:v \in Int(T))$ is a $T$-ensemble of weight functions such that each $h_v$ is a fractional hitting set for $\Omega_v$,
then $h_T=\composition_T(h_v:v \in Int(T))$ is a fractional hitting set for $\Omega_T$.  Furthermore, if all of the $h_v$
are boolean valued (so that $h_v$ is a hitting set), then so is $h_T$.
\item If $h$ is any fractional hitting set for $\Omega_T$, then there exists a $T$-ensemble of weight functions
$(h_v:v \in Int(T))$ such that each $h_v$ is a hitting set for $\Omega_v$, and $h \geq \composition_T(h_v:v \in Int(T))$ pointwise.  Furthermore, if $h$ is boolean valued,
then all of the $h_v$ can be chosen to be boolean valued. 

\end{itemize}
\end{myprop}

\begin{proof}
  For both parts of the lemma we first 
prove the case where every leaf in $T$ is distance $2$ from the root, and then use induction to obtain
the general result.  
  
   Recall that $h_T$ assigns to leaf $l$ the product of the values that $h_v$ assigns to the edges along the unique path from the root to $l$. Thus, if all the $h_v$ are boolean valued, then $h_T$ is boolean valued. 
  We now show that, for all $w \in \Omega_T$,
  \begin{equation} \label{eq:treesum1}
  \sum_{l \in L(T)}h_T(l)w(l) \geq 1.
  \end{equation}
   Let $r$ be the root of $T$ and fix $w \in \Omega_T$. Since $\Omega_T = \Omega_r(\Omega_v : v \in C(r))$, it follows that for some choice of $w_r \in \Omega_r$ and $w_v \in \Omega_v$ (for each $v \in C(r)$)   we have
  \[w = w_r\left(w_v : v \in C(r)\right).\]
  Likewise, by assumption
  \[h_T = h_r(h_v : v \in C(r)).\]
  
 For $v \in C(r)$, let $T_v$ be the subtree whose root is $v$. It follows that
 %Then the left hand side of equation (\ref{eq:treesum1}) may be rewritten as
 \begin{align*}
 \sum_{l \in L(T)}h_T(l)w(l) & = \sum_{v \in C(r)}\sum_{l \in L(T_v)} h_r(v)h_v(l)w_r(v)w_v(l) \\
 & = \sum_{v \in C(r)} h_r(v)w_r(v)\left(\sum_{l \in L(T_v)}h_v(l)w_v(l)\right).
 \end{align*}
  %\[ \sum_{v \in C(r)}\sum_{l \in L(T_v)} h_r(v)h_v(l)w_r(v)w_v(l).\]
  %\[ = \sum_{v \in C(r)} h_r(v)w_r(v)\left(\sum_{l \in L(T_v)}h_v(l)w_v(l)\right).\]
  
  Each $h_v$ is a fractional hitting set for $\Omega_v$, thus the inner sums are all at least 1. Therefore, the above is
  \[ \geq \sum_{v \in C(r)} h_r(v)w_r(v).\]
 This, however, is at least 1 because $h_r$ is a fractional hitting set for the hypergraph $\Omega_r$. This proves  (\ref{eq:treesum1}). 

To see the induction step, for view an arbitrary tree $T$ of uniform depth as a composition $T_r(T_v: v \in C(r))$ where $T_r$ is a rooted star. Then let $h_v = \composition_{T_v}(h_u : u \in Int(T_v))$, which by induction will be a fractional hitting set for $\Omega_{v}:=\composition_{T_v}(\Omega_u : u \in Int(T_v))$. We may then ignore the inner structure of the subtrees $T_v$, treating them as rooted stars, which reduces the problem to the depth 2 case already shown.
  
  For the next part we show that, given any $h_T$ which is a hitting set for $\Omega_T$, we can find weight functions $h_r$ and $\{h_v : v \in C(r)\}$ such that $h_T \geq h_r\left(h_v : v \in C(r)\right)$ pointwise and all the $h_v$, and $h_r$ are hitting sets for $\Omega_v$ and $\Omega_r$ respectively. 
  
  The construction is as follows: For each $v \in C(r)$ we define 
  \[h_r(v) := \min\left(1,\min_{w \in \Omega_v} \sum_{l \in L(T_v)} w(l)h_T(l)\right).\] 
  Note that $h_r(v)$ will be boolean valued if $w$ and $h_T$ are. Let $S = \{v \in C(r) : h_r(v) \neq 0\}$. For $v \in S$ and $l \in L(T_v)$ we define $h_v(l) := h_T(l)/h_r(v)$. For $v \notin S$, we set $h_v \equiv 1$. Again, each $h_v$ defined in this way will be boolean valued if $h_T$ is boolean valued.  It is clear by construction that $h_T \geq h_r(h_v : v \in Int(T))$.
  
  By construction, each $h_v$ is a hitting set for $\Omega_v$. This is trivial if $v \notin S$. Otherwise, if $v \in S$ and $w \in \Omega_v$, then
  \begin{align*}
  \sum_{l \in L(T_v)}w(l)h_v(l) & =  \sum_{l \in L(T_v)}\frac{w(l)h_T(l)}{h_r(v)} \\
 & \geq  \frac{\sum\limits_{l \in L(T_v)}w(l)h_T(l)}{\min_v h_r(v)}  
  \geq \frac{\sum\limits_{l \in L(T_v)}w(l)h_T(l)}{\min\limits_{w' \in \Omega_v} \sum\limits_{l \in L(T_v)}w'(l)h_T(l)} \\
  & \geq 1.
  \end{align*}
  
  It remains to show that $h_r$ is a hitting set for $\Omega_r$. Let $w_r \in \Omega_r$ be given. For each $v \in C(r)$, let $w_v \in \Omega_v$ be such that $\min(1,\sum\limits_{l \in L(T_v)} w_v(l) h_T(l)) = h_r(v)$. Define $w_T := w_r(w_v: v \in C(r))$. Note that $w_T \in \Omega_T$ because $\Omega_T = \Omega_r(\Omega_v : v \in C(r))$. In the following analysis recall that $w_r$ is boolean valued by assumption. 
  \begin{align*} \sum_{v \in C(r)} w_r(v)h_r(v) & = \sum_{v \in C(r)}w_r(v) \min\left(1,\sum_{l \in L(T_v)} w_v(l) h_T(l)\right) \\   
   & = \sum_{v \in C(r):  w_r(v) = 1} \min\left(1,\sum_{l \in L(T_v)} w_v(l) h_T(l)\right) \\
   & \geq \min\left(1, \sum_{v \in C(r): w_r(v) = 1}\sum_{l \in L(T_v)}w_v(l) h_T(l)\right) \\
   & \geq \min\left(1, \sum_{l \in L(T)} w_T(l) h_T(l)\right) \\
   & \geq 1.
  \end{align*}
The induction step works as follows. Given a tree $T$ of uniform depth $k$, view it as a composition $T_r(T_v: v \in C(r))$ where each $T_v$ has depth $k-1$. Then decompose $\Omega_T = \Omega_r(\Omega_{T_v} : v \in C(r))$, where each $\Omega_{T_v}$ is the $T_v$-composition of $(\Omega_u : u \in Int(T_v))$. Apply the height $2$ case to get $h_r$ and $h_{T_v}$ with the desired properties. Then continue this process on each of the subtrees $T_v$ until $h_T$ has been fully decomposed. 
 
\end{proof}
   We are now ready to prove Lemma \ref{assemblage comp}, which we repeat for convenience.
\assemblagecomp*
%\begin{lemma}
%\label{assemblage comp}
%Each of the assemblages $\partial\mc{B}$, $\wit{}$, $\witstar{}$ and $\Psi$ are well-behaved.
%\end{lemma}
\begin{proof}
  We prove each part separately. To prove that the minblock assemblage is well behaved, we only prove the case where $T$ is an indexed tree of height 2, that is where $F = f_r\left(f_v : v \in C(r)\right)$ is a composition of boolean functions and $T = T_r(T_v : v \in C(r))$. The general case will then follow by induction (we omit this part as it follows similarly to the induction in Proposition A).
  
  \textbf{The assemblage $\partial\mc{B}$ is well-behaved:}
 
   We will prove the stronger statement:

 \newtheorem*{myclaim}{Claim B}
\begin{myclaim}
\label{min-block claim}
For each input $x$  to $F$,
   \[ \partial\mc{B}_{x}(F) = \composition_T\left(\partial\mc{B}_{x_v}(f_v)\right).\]
  \end{myclaim}

   %In what follows we view blocks as boolean valued weight functions, we say a block $B$ is empty if $B \equiv 0$. 
  
   Let $x$ be an arbitrary input for $F$. Let $(x_v : v \in Int(T))$ be the assignment ensemble induced by evaluating the circuit for $F$ on assignment $x$. Let $B_T$ be a min-block for $F$ at $x$. $B_T$ induces  a boolean valued weight function over $C(r)$ in the following natural way: $B_r(v) := 1$ if and only if there exists a leaf $l \in L(T_v)$ such that $B_T(l) = 1$. Likewise, $B_T$ induces weight functions $B_v$ on the leaves of the subtrees rooted at $v$ for $v \in C(r)$ in the same way, that is $B_v(l) := 1$ if and only if $B_T(l) = 1$. In this way, $B_T = B_r(B_v : v \in C(r))$. It remains to show that each $B_v$ which is not identically 0 is a min block for $f_v$ at $x_v$ and that $B_r$ is a min block for $f_r$ at $x_r$. 
  
  First we show that, for each $v \in C(r)$, $B_v$ is a block at $x_v$. Suppose for contradiction that $f_v(x_v \oplus B_v) = f_v(x_v)$ for some $v \in C(r)$ where $B_v$ is not empty. Then changing $B_T$ to be $0$ on the leaves of $v$ will create a strictly smaller block $B_T'$. Similarly, if $B_v$ is a block but not a min-block, then again $B_T$ may be modified to be strictly smaller. Thus each $B_v$ is a min block for $f_v$ for $v \in C(r)$. 
  
  Now we show that $B_r$ is a min block. Recall that we defined $B_r(v) =1$ if and only if $B_v$ is not identically $0$. Furthermore, we just showed that if $B_v$ is not identically $0$, then it is a min block. Therefore, it follows that $B_r(v) = 1$ if and only if $f_v(x_v \oplus B_v) \neq f_v(x_v)$. This implies that $F(x \oplus B_T) = f_r(x_r \oplus B_r)$. Thus, since $B_T$ is a block for $F$, $B_r$ must be a block for $f_r$ at $x_r$. However $B_r$ must also be a min block, otherwise, by replacing it by a strictly smaller block $B_r'$, the block $B_r'(B_v : v \in C(r))$ will be strictly smaller than $B_T$. 
  
  We have shown that each min-block $B_T$ may be decomposed as a composition of min-blocks. By a similar argument, if $B_T$ is a composition of min-blocks, then it is a min-block for $F$. This shows that $\partial\mc{B}$ is well-behaved.
  
  \textbf{The assemblages $\wit{}, \witstar{}$ are well-behaved:} 
 
   Let $T$ be an indexed tree and let $(f_v:v \in Int(T))$ be a boolean function ensemble with composition $F$. Let $\mc{B}_T$ be the set of min blocks for $F$ at input $x$. For each $v \in Int(T)$, let $\mc{B}_{v}$ be the set of min blocks for $f_v$ at input $x_v$. We just proved that $\mc{B}_T = \composition_T(\mc{B}_v:v \in Int(T))$. By the second part of Proposition A, for each $h_T$ which is a fractional hitting set for the hygergraph $\mc{B}_T$, there exists a composed fractional hitting set $h_T' = \composition_T(h_v:v \in Int(T))$ such that $h_T' \leq h_T$ (pointwise) and each $h_v$ is a fractional hitting set for $\mc{B}_v$. Thus, each minimal hitting set may be decomposed as a composition of hitting sets. This proves property (2). For property (1), assume that $h_v$ is a fractional hitting set for $\mc{B}_v$ for each $v \in Int(T)$. It follows by Proposition A that $\composition_T(h_v: v \in Int(T))$ is a fractional hitting set for $\mc{B}_T$. This shows that $\witstar{}$ is well-behaved. The same argument works for boolean valued hitting sets, thus $\wit{}$ is well-behaved. 
  
  \textbf{The assemblage $\Psi$ is well-behaved:} If $x$ is an input to $F$, then $\Psi_x(F)$ consists of a single set (the set of sensitive indices) and this set will precisely be the composition of the sets $\Psi_{x_v}(f_v)$.

\end{proof}

%The notion of $T$-composition also gives rise to a dual notion of {\em $T$-decomposition}.
%Given an object $\alpha$  of type $\tau$ over $L(T)$ then a $T$-ensemble 
% $(\alpha_v:v \in Int(T))$ of objects  of type $\tau$ is said to be a $T$-decomposition of %$\alpha$ if
%$\alpha=\composition_T(\alpha_v:v \in Int(T))$. 

\uniformcomp*

\begin{proof}

The first part follows from the second part of Proposition \ref{T-labeling}. For the second part of the present lemma, 
 we first prove the case of $k=2$, the full statement will follow by induction.   Let $F = f_1 \composedwith f_2$ where each $f_i$ is a function on rooted star $T_i$, and let $T = T_1\circ T_2$. Let $\vec{\alpha} = (\alpha^0,\alpha^1)$ be an $m$-optimal selector for the function $F$. Let $(b^c(v) : v \in Int(T))$ be the boolean $T$-labeling induced by evaluating $F(\alpha^c)$  and let $(b_v^c : v \in Int(T))$ be the corresponding assignment ensemble.
  
  We choose $\vec{\beta}_2$ be any $m$-optimal selector for the function $f_2$, and set $\vec{\beta}_1:=(b_r^0,b_r^1)$. Our goal is to prove that $\vec{\beta}= \vec{\beta}_1 \composedwith \vec{\beta}_2$ is also an $m$-optimal selector for $F$.
  
  Writing $\vec{\beta}$ as  $(\beta^0,\beta^1)$, we prove that $m_{\beta^0}(F) \geq m_{\alpha^0}(F)$; 
the analogous result for $\beta^1$ follows similarly. The construction of $\beta^0$ also induces an assignment ensemble which we denote as $(\beta_v^0 : v \in Int(T))$.  Fix a minimum size weight function $w \in \mc{A}_{\beta^0}(F)$. We show how to modify $w$ to obtain a weight function $w' \in \mc{A}_{\alpha^0}(F)$ of size at most $|w|$. This will prove the lemma, since then $m_{\alpha^0}(F)\leq |w'|\leq |w| = m_{\beta^0}(F)$.

Since $\mc{A}$ is well-behaved and $w$ is minimal, we know that $w = w_r\left(w_v : v \in C(r)\right)$ for some choices of $w_r \in \mc{A}_{\beta_r^0}(f_1)$ and $w_v \in \mc{A}_{\beta^0_v}(f_2)$. Note that for each $v$, $f_2(\alpha^0_v) = f_2(\beta^0_v)$. It follows that $m_{\alpha^0_v}(f_2) \leq m_{\beta^0 _v}(f_2)$ (because $\vec{\beta}_2$ is an $m$-optimal selector for $f_2$). Hence, we can find $\rho_v \in \mc{A}_{\alpha^0_v}(f_2)$ such that $|\rho_v| \leq |w_v|$. Having found the functions $\rho_v$, we set $w' := w_r\left(\rho_v: v \in C(r)\right)$ which will be an element in the assemblage $\mc{A}_{\alpha_0}(F)$. Moreover,
\begin{align*}
|w'| = \sum_{v \in C(r)}w_r(v) |\rho_v|
\leq \sum_{v \in C(r)}w_r(v) |w_v| = |w|.
\end{align*}

 To complete the proof for general $k$ we view a composed function $F = f_1 \circ \cdots \circ f_k$ as $f_1 \circ F_{k-1}$, where $F_{k-1}= f_2 \circ \cdots \circ f_k$. By induction on $k$, there is an $m$-optimal selector for $F_{k-1}$ of the form $\vec{\alpha}=\vec{\alpha}_2 \circ \cdots \circ \vec{\alpha}_k$. We may then repeat the proof of the case of height 2, where we view $f_2$ as the function $F_{k-1}$ and choose $\vec{\beta}_2 := \vec{\alpha}$. Then there is an $m$-optimal selector for $F$ which is of the form $\vec{\beta}_1 \circ \vec{\beta}_2 = \vec{\beta}_1 \circ \vec{\alpha}_2 \circ \cdots \vec{\alpha}_k$.

\end{proof}

We now turn to the proof of Lemma \ref{lem:hitting decomp}.  Recall that an AW-selector pair over index set $I$ is
a pair $(\vec{\alpha},\vec{w})$ consisting of an 
assignment selector $\vec{\alpha}$ over $I$ and a weight function selector $\vec{w}$ over $I$.  
We need to define the {\em uniform composition of AW-selector pairs}. 
Let $\vec{\alpha_1}, \cdots, \vec{\alpha_k}$ be assignment selectors and $\vec{w_1}, \cdots, \vec{w_k}$ be weight function selectors over $L(T_1), \cdots, L(T_k)$ respectively. We define $(\vec{\alpha_1}, \vec{w_1}) \circ \cdots \circ (\vec{\alpha_k}, \vec{w_k})$ to be the pair $(\vec{\alpha},\vec{w})$ where $\vec{\alpha}$ is the assignment selector $\vec{\alpha_1} \circ \cdots \circ \vec{\alpha_k}$, and $\vec{w} = (w^0,w^1)$ is a weight function selector defined in the following manner.  Each of $w^0$
and $w^1$ are defined, respectively, as compositions of weight function $T$-ensembles $w^0_T$ and $w^1_T$.  
To construct these ensembles,  first
recall from Section \ref{Boolean labelings} that each component $\alpha^c$ (for $c \in \{0,1\}$)
of the composition $\vec{\alpha}$ is naturally associated to a
 boolean $T$-labeling $b^c_T=(b^c(v) : v \in Int(T))$.  We use the labeling $b^c_T$ to define  the ensemble 
$w_T^c=(w^c_v : v \in Int(T))$ where for  node $v$ is at level $m$ (treating the root as level 1), 
the function $w_v^c$ is a copy of either
$w_m^0$ or $w_m^1$  depending on whether $b^c(v)=0$ or 1.

\hittingdecomp*

\iffalse
\begin{lemma}
Let $f_1,\ldots,f_k$ be a sequence of boolean functions and let $F$ be their composition.  Let $\vec{\alpha}_1,\ldots,\vec{\alpha}_k$ be assignment selectors,
where $\vec{\alpha}_i$ is $f_i$-compatible, and let $\vec{\alpha}$ be their composition.

\begin{itemize}
\item  For any sequence $\vec{w}_1,\ldots,\vec{w}_k$ of weight function selectors, where $(\vec{\alpha}_i,\vec{w}_i)$ is $(f_i,\mc{A})$ compatible,
the composition $(\vec{\alpha},\vec{w}):= (\vec{\alpha}_1,\vec{w}_1) \composedwith \cdots \composedwith (\vec{\alpha}_k,\vec{w}_k)$ is $(F,\mc{A})$-compatible.
\item
There are weight function selectors $\vec{w}_1, \cdots, \vec{w}_k$ such that the weight function $\vec{w}$ that comes from the composition
 $(\vec{\alpha},\vec{w}):=(\vec{\alpha_1},\vec{w}_1) \circ \cdots \circ (\vec{\alpha}_k,\vec{w}_k)$  is an $m$-optimal weight function selector for $F$ at $\vec{\alpha}$.
\end{itemize}
\end{lemma}
\fi

\begin{proof}
In the proof of both statements let $T = T_1 \circ \cdots \circ T_k$ be the indexed tree for the function $F$. Also let $(b(v) : v \in T)$ be the boolean $T$-labeling induced by $\alpha^0$, and let $(\alpha^0_v : v \in Int(T))$ be the corresponding assignment ensemble.

 For the first part, it follows from Lemma \ref{lem:uniform comp} that $\vec{\alpha}$ is $F$-compatible.  Recall the construction of $\vec{w} = (w^0, w^1)$; in particular the weight function $w^0$ is the composition of the weight function ensemble $(w_v: v \in Int(T))$ where, if $v$ is at depth $i$ ($i=1$ being the root), then $w_v$ is a copy of $w_i^{b(v)}$. Also, the assignment $\alpha^0_{v}$ is the assignment $\alpha_i^{b(v)}$. Since $(\vec{\alpha}_i, \vec{w}_i)$ is a compatible pair for each $i$, it follows that $w_v \in \mc{A}_{\alpha^0_v}(f_i)$ for each $v$. Because $\mc{A}$ is well-behaved, we have $w^0 = \composition_T(w_v: v \in Int(T)) \in \mc{A}_{\alpha^0}(F)$. 
   The exact same proof shows that $w^1 \in \mc{A}_{\alpha^1}(f^{(k)})$. This proves that $(\vec{\alpha},\vec{w})$ is an $F$-compatible pair.

   Now we prove the second statement. Again we prove the case $k=2$, letting the general case follow by induction. Let $\vec{\alpha} = (\alpha^0,\alpha^1)$. To construct $\vec{w}$, we will choose $\vec{w}_2$ to be any $m$-optimal weight function selector for $f_2$ at $\vec{\alpha}_2$ and construct $\vec{w}_1 = (w_1^0, w_1^1)$.
  
  We first construct $w_1^0$. Let $w^*$ be any minimum sized weight function in $\mc{A}_{\alpha^0}(f)$. Because $\mc{A}$ is well-behaved and $w^*$ is minimal, we may decompose $w^* = w^*_r(w^*_v: v \in C(r))$ where $w^*_r \in \mc{A}_{\alpha^0_r}(f_1)$ and $w^*_v \in \mc{A}_{\alpha^0_v}(f_2)$. We will set $w_1^0 := w^*_r$ and check that it satisfies the properties we need. Consider $w':= w^*_r(w_v : v \in C(r))$ where $w_v:= w_2^{b(v)}$. Note that $w' \in \mc{A}_{\alpha^0}(f)$ because $\mc{A}$ is well-behaved and moreover it has size
  \[
   |w'| = \sum_{v \in C(r)} w^*_r(v)|w_2^{b(v)}| \leq \sum_{v \in C(r)}w^*_r(v)|w^*_v| = |w^*|. \]
  Here the inequality follows from the fact that $w_2^0$ and $w_2^1$ have minimum sizes in the families $\mc{A}_{\alpha_2^0}(f_2)$ and $\mc{A}_{\alpha_2^1}(f_2)$ respectively. 
  
  In the same manner construct $w_1^1$. Finally, set $(\vec{\alpha},\vec{w}):=(\vec{\alpha}_1,\vec{w}_1) \composedwith (\vec{\alpha}_2,\vec{w}_2)$ where $\vec{w} = (w^0,w^1)$. Then by construction, $w^0 = w'$, and we have shown $|w^0| \leq |w^*|$. Thus, $w^0$ must have minimum size in the family $\mc{A}_{\alpha^0}(f)$. By the same argument, the function $w^1$ will have minimum size in the family $\mc{A}_{\alpha^1}(f)$. Therefore, $\vec{w}$ is an $m$-optimal selector for $f$ at  $\vec{\alpha}$ as desired.

To see the induction step, view $F = f_1 \circ \cdots \circ f_k$ as a composition of two functions $f_1 \circ F_{k-1}$ where $F_{k-1} = f_2 \circ \cdots \circ f_k$. We now use the same construction, only noting that by induction on k we may choose $\vec{w}_2$ to a composition of AW selector-pairs.
\end{proof}

We now show that multiplication of profile matrices encapsulates crucial information about AW selector-pair composition. The following definitions will be helpful. 

\begin{itemize}
\item {\em Profile vector of a weight function $w$ on assignment $x$}. 
The {\em profile} of $(x,w)$ is the pair $p_x(w):=(p_0,p_1)$  where
$p_0:= \sum\limits_{i:x_i=0} w_i$ and $p_1:=\sum\limits_{i:x_i=1}w_i$. 
\item {\em Profile vector family $P_x(f)$ for the assignment $x$ and assemblage $\mc{A}$.} 
  This is the set of all profile vectors $p_x(w)$ where $w$ ranges over weight functions in $\mc{A}_x(f)$.
\end{itemize}

 For any profile matrix $M:=M_{\vec{\alpha},\vec{w}}$, the first row of $M$ is the profile vector $p_{\alpha^0}(w^0)$ and the second row is the profile vector $p_{\alpha^1}(w^1)$.

\matrixmult*
\iffalse
\begin{proposition} %\label{prop:matrix mult}
  Let $(\vec{\alpha}_1, \vec{w}_1), \cdots, (\vec{\alpha}_k, \vec{w}_k)$ be compatible pairs and let $(\vec{\alpha},\vec{w})$ be their composition. Then we have the following relation between profile matrices:
  \[M_{\vec{\alpha},\vec{w}} = M_{\vec{\alpha}_1,\vec{w}_1} M_{\vec{\alpha}_2,\vec{w}_2}\cdots  M_{\vec{\alpha}_k,\vec{w}_k}.\]
  \end{proposition}
\fi

  \begin{proof}
    We prove the special case where $k=2$, the general case will then follow by induction. Let $(\vec{\alpha},\vec{w}) = (\vec{\alpha}_1,\vec{w}_1) \circ (\vec{\alpha}_2,\vec{w}_2)$ be the composed pair, where $\vec{w} = (w^0, w^1)$ and $\vec{\alpha} = (\alpha^0, \alpha^1)$, and let $T$ be the corresponding indexed tree. Let
    \[ M_{\vec{\alpha}_1,\vec{w}_1} = \begin{bmatrix} a_{00} & a_{01} \\ a_{10} & a_{11} \end{bmatrix} \ \ \ \ \ \ \ \  M_{\vec{\alpha}_2,\vec{w}_2} = \begin{bmatrix} b_{00} & b_{01} \\ b_{10} & b_{11} \end{bmatrix} \ \ \ \ \ \ \ \  M_{\vec{\alpha},\vec{w}} = \begin{bmatrix} c_{00} & c_{01} \\ c_{10} & c_{11} \end{bmatrix}.\]
    
    We check that $c_{00} = a_{00}b_{00} + a_{01}b_{10}$, the other entries will follow by similar arguments. We check this by computing the profile vector for the input $\alpha^0$. Let $(b(v) : v \in T)$ be the boolean $T$-labeling induced by the construction of $\alpha^0$ and let $(b_v : v \in Int(T))$ be the corresponding assignment ensemble. By construction, $b_r = \alpha_1^0$ and $b_v = \alpha_2^{b(v)}$ for $v \in C(r)$. Recall that $w^0:= w_r(w_v : v \in C(r))$, where $w_r = w_1^0$ and $w_v:= w_2^{b(v)}$ for $v \in Int(T)$.
    
    The profile vector $p_{\alpha^0}(w^0):= [c_{00}, c_{01}]$. In particular, $c_{00} = \sum_{l \in L(T) \ : \  \alpha^0(l) = 0}w^0(l)$. %Likewise, the profile vectors $p_{\alpha_1^0}(w_1^0) = [a_{00}, a_{01}]$, $p_{\alpha_1^1}(w_1^1) = [a_{10}, a_{11}]$, $p_{\alpha_2^0}(w_2^0) = [b_{00}, b_{01}]$, and $p_{\alpha_2^1}(w_2^1) = [b_{10}, b_{11}]$. 
    Thus, we have
    \begin{align*}
     c_{00} & = \sum\limits_{\substack{l \in L(T) \\ \alpha^0(l) = 0}}w^0(l) \\
     & = \sum\limits_{v \in C(r)} \sum\limits_{\substack{l \in C(v) \\ \alpha^0(l) = 0}}w_r(v)w_v(l) \\
     & = \sum\limits_{\substack{v \in C(r) \\ b_r(v) = 0}} w_r(v)\sum\limits_{\substack{l \in C(v) \\ b_v(l) = 0}} w_2^0(l) + 
     \sum\limits_{\substack{v \in C(r) \\ b_r(v) = 1}} w_r(v)\sum\limits_{\substack{l \in C(v) \\ b_v(l) = 0}} w_2^1(l)  \\
     & = \sum\limits_{\substack{v \in C(r) \\ \alpha_1^0(v) = 0}} w_r(v)\sum\limits_{\substack{l \in C(v) \\ \alpha_2^0(l) = 0}} w_2^0(l) + 
     \sum\limits_{\substack{v \in C(r) \\ \alpha_1^0(v) = 1}} w_r(v)\sum\limits_{\substack{l \in C(v) \\ \alpha_2^1(l) = 0}} w_2^1(l)  \\
     & = \sum\limits_{\substack{v \in C(r) \\ \alpha_1^0(v) = 0}} w_1^0(v)b_{00} + 
     \sum\limits_{\substack{v \in C(r) \\ \alpha_1^0(v) = 1}} w_1^0(v)b_{10} \\ 
         & = a_{00}b_{00} + a_{01}b_{10}.    
    \end{align*}
    %Note that $\sum\limits_{l \in C(v) \ : \ \alpha_2^0(l) = 0} w_2^0(l)$ is the total weight $w_2^0$ assigns to the $0$ indices of the assignment $\alpha_2^0$, thus this sum is equal to $b_{00}$. Likewise, the sum $\sum\limits_{l \in C(v) \ : \ \alpha_2^1(l) = 0} w_2^1(l)$ is equal to $b_{10}$. Thus
    % Here the last equality follows from the fact that $\sum\limits_{v \in C(r) \ : \  \alpha_1^0(v) = 0} w_r(v)$ is the total weight $w_1^0$ assigns to the $0$ indices of $\alpha_1^0$ (this is the definition of $a_{00}$). Likewise, $\sum\limits_{v \in C(r) \ : \ \alpha_1^0(v) = 0} w_r(v)$ is the total weight $w_1^0$ assigns to the $1$ indices of $\alpha_1^1$, which is equal to $a_{01}$.
  \end{proof}  

\iffalse 
Lemma \ref{lem:m classified} (see previous section) and Proposition \ref{prop:matrix mult} imply
\begin{corollary}
\label{cor:m(f) via matrices}
Let $F=f_1 \composedwith \cdots \composedwith f_k$.  Then $m(F)$ is equal to the maximum over all sequences $\vec{\alpha}_1,\ldots,\vec{\alpha}_k$ where $\vec{\alpha}_i$ is an input selector
for $f_i$ of the minimum over all choices of matrices $M_1,\ldots,M_k$ where $M_i$ belongs to the profile matrix family $\mc{M}_{\vec{\alpha}}(f_i)$
of the maximum row sum of the product $M_1 \cdots M_k$.
\end{corollary}
\fi

We now present the two main lemmas which imply our main result. The proofs reduce the statements to two claims regarding the largest eigenvalue of the product of certain matrices which we delay for the next section.

\lowerbound*
\iffalse
\begin{lemma} %\label{lem:m lower bound}
  For any boolean function $f$ and $k \in \mathbb{N}$
   \[m(f^{(k)}) \geq \hat{m}(f)^k/2.\]
\end{lemma}
\fi
 \begin{proof}
 Let $\lambda := \hat{m}(f)$. Let $\vec{\beta}$ be an assignment selector for which $\hat{m}_{\vec{\beta}}(f) = \lambda$. Let $\vec{\alpha} := \vec{\beta}^{(k)}$. By Corollary \ref{cor:m(f) via matrices} and Fact \ref{fact:matrix rho}, the claim will follow from showing that, for any sequence of profile matrices $M_i \in \mc{M}_{\vec{\beta}}(f^{(k)})$, we have
  \begin{equation} \label{eq:meq}
    \rho(M_1 M_2 \cdots M_k) \geq \lambda^k.
  \end{equation} 
  
   Let $\{M_i\}_{i=1}^k$ be any such sequence of matrices and let $M$ be their product. Because of our choice of $\vec{\beta}$, we know that $\rho(M_i) \geq \lambda$ for each $i$. In general, this is not enough to guarantee that $\rho(M) \geq \lambda^k$. However, these matrices contain additional structure which will allow us to make such a conclusion.
   
   Recall that each profile matrix $M_i \in \mc{M}_{\vec{\beta}}(f)$ corresponds to a weight function selector $\vec{w}_i$ which is $\vec{\beta}$-compatible. For each $i,j \in [k]$ let $M_{ij}$ denote the matrix who's first row is the first row of $M_i$ (i.e., the profile vector $p_{\beta^0}(w_i^0)$), and who's second row is the second row of $M_j$ (i.e., the profile vector $p_{\beta^1}(w_j^1)$). Then $M_{ij}$ is precisely the profile matrix $M_{\vec{\beta},\vec{w}_{ij}}$ where $\vec{w}_{ij} = (w_i^0,w_j^1)$. In particular, each $M_{ij} \in \mc{M}_{\vec{\beta}}(f)$ and $\rho(M_{ij}) \geq \lambda$ by the definition of $\lambda$. Noting this property, we apply Lemma \ref{lemma_prod_supermult} (see the following section \ref{sec:matrix}) and conclude that 
  \[ \rho(M) \geq \lambda^k.\]
  
 \end{proof}

\upperbound*
\iffalse
\begin{lemma} %\label{lem:m upper bound}
   For any boolean function $f$ on $n$ variables and $k \in \mathbb{N}$,
     \[m(f^{(k)}) \leq 2n k\hat{m}(f)^{k-1}.\]
 \end{lemma}
\fi
\begin{proof}
   Let $\lambda := \hat{m}(f)$. Take $\vec{\alpha}$ which is an $m$-optimal selector for $f^{(k)}$. By lemma \ref{lem:uniform comp}, we may assume that $\vec{\alpha} = \vec{\alpha_1} \circ \vec{\alpha_2} \circ \cdots \circ \vec{\alpha_k}$. By Corollary \ref{cor:m(f) via matrices}, the claim will follow by exhibiting matrices $M_i \in \mc{M}_{\alpha_i}(f)$ such that
  \[||M_1 M_2 \cdots M_k||_\infty \leq nk\lambda^{k-1}.\]
  
     Let $\mc{U}_i$ be the profile vector family  $P_{\alpha_{i}^0}(f)$, 
and let $\mc{V}_i = P_{\alpha_{i}^1}(f)$. When considering the possible choices of $M_i \in \mc{M}_{\vec{\alpha}_i}(f)$, the set $\mc{U}_i$ is the set of possible first rows of $M_i$. Likewise, $\mc{V}_i$ is the set of possible second rows of $M_i$. One may hope to use the definition of $\lambda$ and choose each $M_i$ such that $\rho(M_i) \leq \lambda$. This in general though is not enough to bound all entries in the product $M_1 \cdots M_k$. Once again we need to use additional structure of these matrix families. Note crucially that, for each $i,j \in [k]$, there exists $u \in \mc{U}_i$ and $v \in \mc{V}_j$ such that \[\rho\left(\left[ \begin{array}{cc}
      u \\ v
    \end{array} \right]\right) \leq \lambda.\] 
    This follows by the definition of $\lambda$ and the fact that the set of profile matrices 
    \[\mc{M}_{\vec{\alpha}_{ij}}(f) = \{ \left[ \begin{array}{cc}
      u \\ v
    \end{array} \right] \mid u \in \mc{U}_i, v\in \mc{V}_j \},\]
    where $\vec{\alpha}_{ij} := (\alpha_{i}^0, \alpha_{j}^1)$.
     By Corollary \ref{cor:matixinfty} (see section \ref{sec:matrix}), there exists matrices $M_1, M_2, \cdots, M_k$, where $M_i \in \mc{M}_{\vec{\alpha}_i}(f)$ for each $i$, such that $||M_1M_2\cdots M_k||_\infty \leq nk\lambda^{k-1}$. 
    
   \end{proof}

\subsection{Facts about non-negative matrices} \label{sec:matrix}

%For any square matrix $M$, we denote by $\rho(M)$ the quantity $\max\{|\lambda|\ |\ \text{$\lambda$ an eigenvalue of $M$}\}$.
In this subsection, we prove Lemmas \ref{lemma_prod_supermult} and \ref{lemma_prod_submult} which
were used in the previous subsection.
We will need the following well-known facts about $2\times 2$ non-negative matrices that follow from Perron-Frobenius theory (for omitted proofs see, e.g., \cite[Chapter 8]{meyer}).

\begin{fact}
\label{fact_non_neg}
Fix $A\in \real^{2\times 2}_{\geq 0}$. We have the following:
\begin{enumerate}
  \item There exists a non-zero $z\geq 0$ s.t. $Az = \rho(A)z$.
	\item For $\lambda\in\real$, the following are equivalent: (1) $\rho(A) \geq \lambda$, (2) $\exists x \geq 0$ such that $x\neq 0$ and $Ax\geq \lambda x$, and (3)  For every $\varepsilon>0$, there exists an $x > 0$ such that $Ax \geq (\lambda-\varepsilon)x$.
	\item For $\lambda \in \real$, we have $\rho(A)\leq \lambda$ iff for every $\varepsilon>0$, there is an $x > 0$ s.t. $Ax\leq (\lambda + \varepsilon)x$.
	\item $\lVert A\rVert_{\infty}\geq \rho(A)/2$.
	\item $\lim_{k\rightarrow \infty} \lVert A^k\rVert_{\infty}^{1/k} = \rho(A)$.
\end{enumerate}
\end{fact}

We can now prove the two main lemmas of this subsection.

\begin{lemma}
\label{lemma_prod_supermult}
Let $M_1,\ldots,M_k\in\real^{2\times 2}_{\geq 0}$ and let $M := M_1\cdots M_k$. For each $i,j\in [k]$, let $M_{i,j}$ denote the matrix whose first and second rows are the first row of $M_i$ and the second row of $M_j$ respectively. If $\rho(M_{i,j})\geq \lambda \geq 0$ for each $i,j\in [k]$, then $\rho(M)\geq \lambda^k$.
\end{lemma}

\begin{proof}
The lemma is trivial for $\lambda=0$. Thus, we assume that $\lambda > 0$.
By dividing each matrix through by $\lambda$, we can assume w.l.o.g. that $\lambda = 1$. 
In this case, we need to show that $\rho(M) \geq 1$.

%Given the lemma in this special case, we can handle the general case by scaling each matrix $M_i$ by $1/\lambda$ and deducing that $\rho((1/\lambda^k) M) = (1/\lambda^k)\rho(M)\geq 1$, which is what we need to show. Hence, from now on, we assume that $\lambda=1$.

By Fact \ref{fact_non_neg}, we can show that $\rho(M)\geq 1$ by showing that there exists a non-zero $z\in \real^2$ s.t. $z\geq 0$ and $Mz\geq z$. To do this, it suffices to produce a $z$ as above s.t. $M_iz \geq z$ for each $i$. %This is because given such an $z$, we have $Mz = M_1\cdots M_k z \geq M_1\cdots M_{k-1} z\geq M_1\cdots M_{k-2} z\cdots \geq z$ where the $i$th inequality follows from the fact that $M_{k-i+1}z \geq z$ and all the $M_j$ and $z$ are non-negative.

%Therefore, we need to show that there exists a non-zero $z\in \real^2_{\geq 0}$ s.t. $M_iz \geq z$ for each $i$. 
Denote by $u_i = (u_{i,1},u_{i,2})$ and $v_i=(v_{i,1},v_{i,2})$ the first and second rows (respectively) of $M_i$. We need $M_iz \geq z$, which is the same as requiring that $\ip{u_i'}{z}\geq 0$ and $\ip{v_i'}{z}\geq 0$ for every $i$, where $u_i' = (u_{i,1}-1,u_{i,2})^T$ and $v_i' = (v_{i,1},v_{i,2}-1)^T$. Clearly, if $u_i'$ or $v_i'$ is non-negative, the corresponding constraint is trivial (since we are looking for $z\geq 0$). Let $P$ and $Q$ denote the set of $i$ where $u_{i,1} < 1$ and $v_{i,2} < 1$ respectively. 

Thus the constraint corresponding to $u_i'$ for $i\in P$ may be rewritten as $z_1\leq (u_{i,2}/(1-u_{i,1}))\cdot z_2$. Clearly, this constraint gets strictly harder to satisfy as the parameter $u_{i,2}/(1-u_{i,1})$ gets smaller and therefore, to satisfy all the constraints indexed by $P$, it suffices to satisfy just the constraint corresponding to $i_0\in P$ for which this parameter is minimized. Similarly, there is a $j_0\in Q$ s.t. any non-negative $z$ that satisfies $\ip{v_{j_0}'}{z}\geq 0$ automatically satisfies all the other constraints indexed by $Q$. However, we know that $\rho(M_{i_0,j_0}) \geq 1$ and hence by Fact \ref{fact_non_neg}, there is some non-zero $z\in \real^2_{\geq 0}$ s.t. $M_{i_0,j_0}z\geq z$ and thus $\ip{u_{i_0}'}{z}\geq 0$ and $\ip{v_{j_0}'}{z}\geq 0$. This $z$ satisfies all the constraints and hence has the property that $M_iz\geq z$ for each $i\in [k]$.
\end{proof}

Given $u,v\in\real^2$, we denote by $\left[\begin{smallmatrix} u \\ v\end{smallmatrix}\right]$ the $2\times 2$ matrix whose first and second rows are $u$ and $v$ respectively.

\begin{lemma}
\label{lemma_prod_submult}
Assume we have compact subsets $U_1,\ldots,U_k,V_1,\ldots,V_k\subseteq \real^{2}_{\geq 0}$ s.t. for each $i,j\in [k]$, there exists $u_{i,j}\in U_i$ and $v_{i,j}\in V_j$ s.t. the matrix $\left[\begin{smallmatrix} u_{i,j} \\ v_{i,j}\end{smallmatrix}\right]$ satisfies $\rho\left(\left[\begin{smallmatrix} u_{i,j} \\ v_{i,j}\end{smallmatrix}\right]\right)\leq \lambda$. Then, there exist $u_i\in U_i$ and $v_i\in V_i$ for each $i\in [k]$ s.t.  the matrices $M_i:= \left[\begin{smallmatrix} u_i \\ v_i\end{smallmatrix}\right]$ and $M^{[i,j]}:= M_i\cdot M_{i+1}\cdots M_j$ for $i \leq j\in [k]$ satisfy $\rho(M^{[i,j]})\leq \lambda^{j-i+1}$.
\end{lemma}

\begin{proof}
We will show that for each $\varepsilon>0$, there is a choice of $u_i\in U_i,v_i\in V_i$ ($i\in [k]$) so that for $M_i:= \left[\begin{smallmatrix} u_i \\ v_i\end{smallmatrix}\right]$ and $M^{[i,j]} := M_i\cdots M_j$, we have
\begin{equation}
\label{eq_submult}
\rho(M^{[i,j]})=\rho(M_i\cdots M_j)\leq (\lambda+\varepsilon)^{j-i+1}.
\end{equation}
for each $i,j\in [k]$ with $i<j$. Since the sets $U_i,V_i$ for $i\in [k]$ are all compact and $\rho:\real^{2\times 2}\rightarrow \real$ is a continuous function, a standard argument shows there must be a choice of these vectors so that $M$ as defined above in fact satisfies the requirements of the lemma. 

Fix $\varepsilon > 0$ and let $\lambda' = \lambda+\varepsilon$. We first show how to choose $u_i,v_i$ ($i\in [k]$) and $z\in \real^{2}_{>0}$ such that for each $i$, $M_i := \left[\begin{smallmatrix} u_i \\ v_i\end{smallmatrix}\right]$ satisfies $M_i z \leq \lambda' z$. We then show how this implies (\ref{eq_submult}).

\begin{claim}
\label{claim_choosing_certs}
There exist $u_i\in U_i$ and $v_i\in V_i$ for each $i\in [k]$ and a $z\in \real^2_{>0}$ such that for each $i\in [k]$, we have $M_iz\leq \lambda' z$, where $M_i$ is as defined above. 
\end{claim}

\begin{proof}[Proof of Claim \ref{claim_choosing_certs}]
The vectors $u_1,\ldots,u_k,v_1,\ldots,v_k$ and $z\in \real^2_{>0}$ that we choose will in fact have the stronger property that for each $i,j\in [k]$, we will have $M_{i,j}z \leq z$, where $M_{i,j}:= \left[\begin{smallmatrix} u_i \\ v_j\end{smallmatrix}\right]$. Let us fix $i,j\in [k]$ and consider the problem of coming up with such a $u_i,v_j$, and $z$. Therefore, we want $u_i\in U_i$ and $v_j\in V_j$ s.t. 
\begin{align*}
\ip{u_i}{z} \leq \lambda' z_1\qquad\qquad
\ip{v_j}{z} \leq \lambda' z_2
\end{align*}

We can rewrite the above constraints on $z$ as
\begin{align*}
\ip{u_i'}{z} \leq 0\qquad\qquad\qquad
\ip{v_j'}{z} \leq 0
\end{align*}
where $u_i' := (\lambda'-u_{i,1},u_{i,2})$ and $v_j' := (v_{j,1},\lambda'-v_{j,2})$. Consider the set of constraints $\{\ip{u_i'}{z} \leq 0\ |\ u_i\in U_i\}$. Note that this set of constraints has the property is that there is a \emph{weakest constraint}: more precisely, there exists a $u_i\in U_i$ s.t. for any $z>0$, if there exists a $\overline{u}_i\in U_i$ s.t. $\ip{\overline{u}_i'}{z} \leq 0$, then $\ip{u_i'}{z}\leq 0$ as well. Similarly, we also have a $v_j\in V_j$. 

We need a crucial observation regarding the vectors $u_i,v_i$ chosen above. By the assumptions of Lemma \ref{lemma_prod_submult}, for every $i,j\in [k]$, we know that for each $i,j\in [k]$, there is \emph{some} choice of $u_{i,j}\in U_{i}$ and $v_{j,i}\in V_{j}$ so that $\rho\left(\left[\begin{smallmatrix} u_{i,j} \\ v_{i,j}\end{smallmatrix}\right]\right)\leq \lambda$. By Fact \ref{fact_non_neg}, this means that there is some $z_{i,j}\in \real^{2}_{>0}$ s.t. $\left[\begin{smallmatrix} u_{i,j} \\ v_{i,j}\end{smallmatrix}\right]z_{i,j} \leq \lambda' z_{i,j}$, which is equivalent to saying that $\ip{u_{i,j}'}{z_{i,j}}\leq 0$ and $\ip{v_{i,j}'}{z_{i,j}}\leq 0$. But this implies that $\ip{u_i}{z_{i,j}}\leq 0$ and $\ip{v_{j}}{z_{i,j}}\leq 0$ as well. Thus, we have shown that

\begin{observation}
\label{obs_submult}
For every $i,j\in [k]$, there exists a $z_{i,j}\in\real^2_{>0}$ s.t. $\ip{u_i'}{z_{i,j}}\leq 0$ and $\ip{v_j'}{z_{i,j}}\leq 0$.
\end{observation}

Now that we have chosen $u_i,v_i$ for each $i\in [k]$, we only need to choose $z\in \real^2_{>0}$ as mentioned above. Again, we need to choose $z\in \real^2_{>0}$ so that for each $i,j$, $\ip{u_{i}'}{z}\leq 0$ and $\ip{v_{j}'}{z}\leq 0$. Consider the sets of constraints $\{\ip{u_{i}'}{z}\leq 0\ |\ i\in [k]\}$ and $\{\ip{v_j'}{z}\leq 0\ |\ j\in [k]\}$. This time we consider the \emph{strongest constraints} in these sets: in other words, we fix an $i_0\in [k]$ so that for any $z>0$, if $\ip{u_{i_0}'}{z}\leq 0$, then in fact $\ip{u_{i}'}{z}\leq 0$ for every $i\in [k]$ and a $j_0\in [k]$ similarly for the $v_j$. By Observation \ref{obs_submult}, we know that there is a $z:=z_{i_0,j_0}>0$ that satisfies these constraints and since these are the strongest constraints, we see that $z$ satisfies $M_{i,j}z\leq z$ for every $i,j\in [k]$. 
\end{proof}

Fix any $i,j\in [k]$ s.t. $i<j$ and consider $M^{[i,j]} = M_i\cdots M_j$, where the $M_\ell$ ($\ell\in [k]$) are as given by Claim \ref{claim_choosing_certs}. We show $\rho(M^{[i,j]})\leq (\lambda')^{j-i+1}$.  By Fact \ref{fact_non_neg}, it suffices to obtain $z\in\real^2_{>0}$ s.t. $M^{[i,j]}\cdot z \leq (\lambda')^{j-i+1} z$. Consider the $z$ guaranteed to us by Claim \ref{claim_choosing_certs}. We have $M^{[i,j]}\cdot z = (M_i\cdots M_j)z \leq (M_i\cdots M_{j-1})(\lambda' z)\ldots \leq (\lambda')^{j-i+1} z$, where the inequalities follows from the choice of $z$ and the fact that the matrices $M_\ell$ are all non-negative. This finishes the proof of Lemma \ref{lemma_prod_submult}.
\end{proof}

\begin{corollary} \label{cor:matixinfty}
  Let $U_1,\ldots,U_k,V_1,\ldots,V_k$, and $\lambda$ be as in Lemma \ref{lemma_prod_submult}. Suppose further that, for any $i$ and any $u \in U_i$ and $v \in V_i$, the entries of $u$ and $v$ are bounded above by a constant $n$. Then there exist $u_i\in U_i$ and $v_i\in V_i$ for each $i\in [k]$ such that the matrices $M_i:= \left[\begin{smallmatrix} u_i \\ v_i\end{smallmatrix}\right]$ satisfy $||M_1 M_2 \cdots M_k||_\infty \leq n k \lambda^{k-1}$.
\end{corollary}
\begin{proof}
  For $i \in [k]$ let $u_i, v_i$ and $M_i := \left[\begin{smallmatrix} u_i \\ v_i\end{smallmatrix}\right]$ be the matrices guaranteed by Lemma \ref{lemma_prod_submult}. Also, let $M^{[1,i]} := M_1 M_2 \cdots M_i$.
  By lemma \ref{lemma_prod_submult}, for each $i$ we have that $\rho(M_i) \leq \lambda$ and $\rho(M^{[1, i]}) \leq \lambda^i$. It is easy to check that given a $2\times 2$ matrix $M$ with non-negative entries, if $\rho(M) \leq C$ then the diagonal entries are both $\leq C$.
  
   We prove by induction on $i$ that the matrices $M^{[1,i]}$ are entry-wise $\leq \left[\begin{smallmatrix}  \lambda^i &  n i \lambda^{i-1} \\
                   n i \lambda^{i-1} & \lambda^i
                  \end{smallmatrix}\right]$.
   The base case follows because we have 
   $M_1 \leq \left[\begin{smallmatrix}  \lambda &  n   \\
                                          n & \lambda
 \end{smallmatrix}\right]$. The diagonal entries of $M^{[1, i+1]}$ are $\leq \lambda^{i+1}$, because $\rho(M^{[1, i+1]}) \leq \lambda^{i+1}$. For the off diagonal entries, note that $M^{[1, i+1]} = M^{[1, i]} M_{i+1}$. By the inductive hypothesis, $M^{[1, i]} \leq \left[\begin{smallmatrix}  \lambda^i &  n i \lambda^{i-1}  \\
n i \lambda^{i-1} & \lambda^i
 \end{smallmatrix}\right]$. Also because $\rho(M_i) \leq \lambda$ we have $M_i \leq\left[\begin{smallmatrix} \lambda &  n   \\
n & \lambda
\end{smallmatrix}\right]$. Thus the off diagonal entries of $M^{[1, i+1]}$ are bounded above by $ n\lambda^i + n i \lambda^{i} = n (i+1) \lambda^i$.

This completes the proof as then, $||M^{[1,k]}||_\infty \leq nk \lambda^{k-1}$.
\end{proof}

\section{The behaviour of Block sensitivity under iterated composition} \label{sec:blocktensor}

In this section, we characterize the behavior of the block sensitivity $bs(f)$ under iterated composition. We show that for any function $f:\{0,1\}^I\rightarrow\{0,1\}$, we have $\alimf{bs}{f} = \alimf{(bs^{\ast})}{f}$. We use similar notation as in the previous sections such as the concepts of indexed trees $T$, $T$-ensembles, and $T$-compositions, only now we will denote by $I$ to be the index set for a function $f$ (which corresponds to a rooted star $T$).

We state the main result of this section formally below.
\begin{theorem}
\label{thm_bs_tensor}
For any boolean function $f:\{0,1\}^I\rightarrow \{0,1\}$, we have $bs^{\lim}(f) = (bs^{\ast})^{\lim}(f)$.
\end{theorem}

The above is easily proved when $f$ is either monotone or anti-monotone. In this case, we know that for each $k\in\naturals$, $f^{(k)}$ is either monotone or anti-monotone and hence $bs(f^{(k)}) = C(f^{(k)})$ \cite{nisan, buhrmandewolf}. As $bs^\ast(f^{(k)})$ is sandwiched between $bs(f^{(k)})$ and $C(f^{(k)})$, we have $bs(f^{(k)}) = bs^{\ast}(f^{(k)})$ and thus we are done. So from now on, we assume that $f$ is neither monotone nor anti-monotone.

\subsection{Some simple claims}

Recall Fekete's lemma for superadditive sequences (see, e.g., 
\cite[Section A.4]{fgt}).
\begin{lemma}[Fekete's lemma]
\label{lemma_fekete}
Let $\{a_m\}_{m\in\naturals}$ be a sequence of real numbers such that for any $p,q\in\naturals$, $a_{p+q}\geq a_p + a_q$. Then, the limit $\lim_{k\rightarrow\infty}a_k/k$ exists (and is possibly infinite) and moreover, we have $\lim_{k\rightarrow\infty}a_k/k = \sup_{k}a_k/k$.
\end{lemma}

We have the following easy corollary to the above lemma for sequences that are ``almost superadditive''.

\begin{corollary}
\label{corollary_fekete}
Let $\{a_m\}_{m\in\naturals}$ be a sequence of real numbers such that for any $p,q\in\naturals$, $a_{p+q}\geq a_p + a_q -c $ for some fixed $c\in\real^{\geq 0}$. Then, the limit $\lim_{k\rightarrow\infty}a_k/k$ exists.
\end{corollary}

\begin{proof}
%Clearly, if the sequence $\{a_m\}_m$ is bounded, then we have $\lim_{k\rightarrow\infty}a_k/k = \limsup_{k\rightarrow\infty}a_k/k = 0$ and we are done. So we assume that $\{a_m\}_m$ is unbounded. In particular, there exists $m$ such that $a_m\geq 2c$ and thus, for all $p\geq m$, we have $a_p \geq a_m + a_{p-m} - c \geq a_m - c \geq c$. Thus, $a_p\geq c$ for all large enough $p$. By removing a finite prefix of this sequence, we may assume that this is true for all $p$. 
Consider the sequence $\{b_m\}_{m\in\naturals}$ defined by $b_m = a_m - c$. Then, $\{b_m\}_m$ is clearly superadditive and moreover, we have $\lim_{k\rightarrow\infty} (a_k - b_k)/k = 0$. Thus, by Lemma \ref{lemma_fekete}, we are done. 
\end{proof}

\begin{lemma}
\label{lemma_rbs_supadd}
Fix any boolean function $f:\{0,1\}^I\rightarrow\{0,1\}$ and any $x\in \{0,1\}^I$. For any $M,k\geq 1$, we have $bs^{kM}_x(f)\geq k\cdot bs^{M}_x(f)$.
\end{lemma}

\begin{proof}
Given any $M$-fold packing $\mc{B}$ of blocks of size $s$ in $\mc{B}_x(f)$, we can construct a $kM$-fold packing of blocks $\mc{B}'$ of size $ks$ in $\mc{B}_x(f)$ by simply repeating $\mc{B}$ $k$ times. When $\mc{B}$ is chosen to be the $M$-fold packing of maximum size for $f$ at $x$, this shows that $bs^{kM}_x(f)\geq k|\mc{B}| = k\cdot bs^M_x(f)$. 
\end{proof}

The following lemma will be crucial in showing that $bs(f^{(k)})$ grows like $bs^\ast(f^{(k)})$.

\begin{lemma}
\label{lemma_bs_to_rbs}
Let $g_i:\{0,1\}^{I_i}\rightarrow\{0,1\}$ ($i\in [2]$) be any non-constant boolean functions. Let $G$ denote the depth-$2$ composition $g_1 \circ g_2$ defined on the index set $I_1 \times I_2$. Then, for any $b\in\{0,1\}$, we have
\[
bs_b(G) \geq bs^{M}_b(g_1)
\]
where $M = \min\{bs_0(g_2),bs_1(g_2)\}$.
\end{lemma}

\begin{proof}
As a short remark, we may view the index set $I_1 \times I_2$ as the leaves of the tree $T = T_1 \circ T_2$ where $T_i$ is a rooted star corresponding to the index set $I_i$. Also, for an assignment $x$ to $I_1$ and $i \in I_1$, we will use $x(i)$ to denote the boolean value $x$ assigns to $i$.

We prove the lemma for $b=0$; an identical proof works for $b=1$.
Let $\vec{\alpha} = (\alpha^0,\alpha^1)$ be a $g_2$-optimal selector (so $bs_{\alpha^b}(g_2) = bs_b(g_2)$ for $b\in \{0,1\}$). Let $x\in g_1^{-1}(0)$ be chosen so that $bs^M_{x}(g_1) = bs^M_0(g_1)$. Consider the composed assignment $X=x\circ \vec{\alpha}$ to the input of $G$. We will show that $bs_X(G)\geq bs^{M}_0(g_1)$, which will prove the lemma.

For $b\in \{0,1\}$, let $\mc{B}_b$ be any maximum-sized packing  in the hypergraph $\mc{B}_{\alpha^b}(g_2)$. Note that $\min\{|\mc{B}_0|,|\mc{B}_1|\}= M$. Let $\mc{B}$ be a maximum-sized $M$-fold block packing in $\mc{B}_{x}(g_1)$. We now give an
algorithm that constructs a block packing $\mc{B}'$ in $\mc{B}_X(G)$ such that $|\mc{B}'| = |\mc{B}| = bs^M_0(g_1)$.

For each $i\in I_1$, the algorithm maintains a  packing $\mc{B}^i$   of the hypergraph $\mc{B}_{\alpha^{x(i)}}(g_2)$.
We initialize $\mc{B}^i$ to be $\mc{B}_{x(i)}$.  We now perform the following for each
block $B \in \mc{B}$ (considered in some arbitrary order):

\begin{itemize}
\item For each $i \in I_1$, define the set $B_i$ to be empty if $i \not\in B$,
and to be a member of $\mc{B}_{x(i)}$ if $i \in B$.  Let $B'$ be the composition
$B(B^i: i\in I_1)$.   (Here
the composition of blocks is {\em subset composition} which, as defined in Section \ref{sec:comp}, is obtained 
by viewing each block
as a boolean weight function, and using composition of weight functions.) 
\item For $i\in B$, the set $\mc{B}^i$ is updated to $\mc{B}^i\setminus \{B^i\}$. 
\end{itemize}

The blocks $B'$ thus constructed are easily seen to belong to $\mc{B}_X(G)$ and to be pairwise disjoint and so form 
a block packing in $\mc{B}_X(G)$.  Provided that we can carry out the process for each block $B \in \mc{B}$
we get the correct number of blocks in our packing.
We need to verify that the first step inside the loop is well-defined, for which we require that when the block $B$ is
considered, for each $i \in B$, $\mc{B}_{(x(i)}$ must be nonempty so that we can select $B_i$.
This is true since $\mc{B}_{x(i)}$ initially has size at least $M$, and decreases by 1 each time
we consider a block $C$ that contains $i$, and $i$ belongs to at most $M$ blocks of $\mc{B}$.

\end{proof}

Lemmas \ref{lemma_rbs_supadd} and \ref{lemma_bs_to_rbs} yield the following.

\begin{corollary}
\label{corollary_bs_infty}
Let $f$ be such that $\min\{bs_0(f),bs_1(f)\}\geq 2$. Then, $\min\{bs_0(f^{(k)}), bs_1(f^{(k)})\}\geq 2^k$. In particular, $\min\{bs_0(f^{(k)}),bs_1(f^{(k)})\}$ goes to infinity as $k\rightarrow \infty$. 
\end{corollary}

\begin{proof}
By Lemmas \ref{lemma_bs_to_rbs} and \ref{lemma_rbs_supadd}, for any $k\geq 1$, we have $bs(f^{(k+1)})\geq bs^2(f^{(k)}) \geq 2 bs(f^{(k)})$. Hence, by induction on $k$, we have the claim. 
\end{proof}

\subsection{Proof of Theorem \ref{thm_bs_tensor}}

Throughout $f:\{0,1\}^{I}\rightarrow\{0,1\}$ is a boolean function defined on index set $I:= [n]$ that is neither monotone nor anti-monotone.

We start off by arguing that $\lim_{k\rightarrow\infty} bs(f^{(k)})^{1/k}$ exists. In order to do this, we need the following simple claim.

\begin{lemma}
\label{lemma_bs0vs1}
Let $f$ be an $n$-variate boolean function that is neither monotone nor antimonotone.  
For all $k \geq 0$ and $b\in \{0,1\}$, we have
\begin{align*}
bs(f^{(k)})&\leq bs_b(f^{(k+1)})\leq n\cdot bs(f^{(k)})\\
bs^\ast(f^{(k)})&\leq bs^\ast_b(f^{(k+1)})\leq n \cdot bs^\ast(f^{(k)})
\end{align*}
(Here $f^{(0)}$ denotes the univariate identity function.)
In particular, for $k\geq 1$, we have $\min\{bs_0(f^{(k)}),bs_1(f^{(k)})\}\geq bs(f^{(k)})/n$ and similarly for the fractional block sensitivity.
\end{lemma}

Note that the hypothesis that $f$ is non-monotone is essential.  If $f$ is the $n$-variate OR function 
then $bs_0(f^{(k)})=n^k$ while $bs_1(f^{(k)})=1$.    The hypothesis that $f$ is not antimonotone
is not essential and is included for convenience.

\begin{proof}
We only prove the claim for block sensitivity. The case of fractional block sensitivity follows by using the exact same reasoning for fractional block packings.

\iffalse
The variables of $f^{(k+1)}$ are indexed by $I^{k+1}$.  
To prove the first inequality, we first note the following easy fact: If $g$ is a restriction of $f$ (i.e. is obtained
from $f$ by fixing some variables) then $bs_b(g) \leq bs_b(f)$ for $b \in \{0,1\}$.

Say that an $n$-variate boolean function $h$ is a rotation of a boolean function $g$ if there is an assignment $a \in \{0,1\}^n$
such that $h(x)
It is not necessarily the case that $f^{(k+1)}$ has a restriction that is 
Now we claim that $f^{(k)}$ and $\bar{f}^{(k)}$ (the complement function) can be
obtained as restrictions of $f^{(k+1)}$.  From this we get that for $b \in \{0,1}$,
$bs_b(f^{(k)}) \leq bs_b(f^{(k+1)}$ and $bs_{1-b}(f^{(k)})=bs_b(\bar{f}^{(k)}) \leq bs_b(f^{k+1})$
as required.  To see that $f^{(k)}$ is a restriction of $f^{(k+1)}$, 
\fi

%View $f^{(k+1)}$ as a function on the leaves of indexed tree $T := T(I)^{(k+1)}$, where $T(I)$ is a rooted star corresponding to the interval $I$ (see Section \ref{subsubsec:indexed trees} for definition of composition of indexed trees). For this proof we view $T = T_r(T_i : i \in I)$.$

We start with the first inequality.  We show it for the case $b=0$, the case $b=1$ is similar.
Let $c \in \{0,1\}$ such that $bs(f^{(k)}) = bs_c(f^{(k)})= N$. Fix assignments $\alpha^0,\alpha^1$  to $I^{(k)}$ so that the selector $\vec{\alpha}:= (\alpha^0,\alpha^1)$ is $f^{(k)}$-compatible and moreover, $\alpha^c$ 
satisfies $bs_{\alpha^c}(f^{(k)}) = N$. 

Since $f$ is neither monotone nor anti-monotone, we can fix an assignment $x$ to $I$ such that $x\in f^{-1}(0)$ and flipping some index $i$ from $c$ to $1-c$ in $x$ results in an assignment $x'\in f^{-1}(1)$. Let $X$ be the assignment to $I^{(k+1)}$ defined by $X = x\circ \vec{\alpha}$. We claim that $bs_X(f^{(k+1)})\geq N$, which will prove the lower bound.

To see this, note that for any block $B$ belonging to $\mc{B}_{\alpha^c}(f^{(k)})$, we can construct a block $\mathrm{lift}(B)\in\mc{B}_{X}(f^{(k+1)})$ defined using composition as $\mathrm{lift}(B):=e_i(B^j: j\in I)$, where $e_i$ is the singleton block $\{i\}$ and $B^j = B$ for $j=i$ and $\emptyset$ otherwise. Using this method, any block packing $\mc{B}$ in $\mc{B}_{\alpha^b}(f^{(k)})$ may be ``lifted'' to a block packing $\mc{B}' = \{\mathrm{lift}(B)\ |\ B\in\mc{B}_{\alpha^b}(f^{(k)})\}$ of the same size as $\mc{B}$. Hence, $bs_X(f^{(k+1)})\geq bs_{\alpha^c}(f^{(k)}) = N$, which proves the first inequality. 

Next we prove the second inequality. 
Let $X$ be an assignment to $I^{(k+1)}$; we want to show that $bs_X(f^{(k+1)}) \leq n \cdot bs(f^{(k)})$.
Let $\mc{B}$ be any block packing in $\mc{B}_X(f^{(k+1)})$ of maximum size. We may assume that $\mc{B}$ contains minimal blocks only, that is, $\mc{B}\subseteq \partial\mc{B}_X(f^{(k+1)})$.

Let $\pi$ denote the mapping from $I^{(k+1)}$ to $I^k$ obtained by mapping $i_1,\ldots,i_{k+1}$ to $i_2,\ldots,i_{(k+1)}$.
For $i \in I$, let $U_i$ be the set of $i_1,\ldots,i_{k+1} \in I^{(k+1)}$ with $i_1=i$ and let $X_i$ be the assignment
to $I^{(k)}$ with $X_i(j_1,\ldots,j_k)=X(i,j_1,\ldots,j_k)$. 
For each block $B \in \mc{B}$, let $B_i=B \cap U_i$.  Since each $B \in \mc{B}$ is a minimal block for $f^{(k+1)}$ at $X$, 
it follows that if $B_i \neq \emptyset$ then
$\pi(B_i)$ is a block for $f^{(k)}$ at $X_i$ (otherwise $B-B_i$ would be a block for $f^{(k+1)}$ at $X$, contradicting
the minimality of $B$).  Let $\mc{B}_i=\{\pi(B_i): B \in \mc{B}, B_i \neq \emptyset\}$.  Then
$\mc{B}_i$  is a packing of blocks for $f^{(k)}$ at $X_i$.  Since for each $B \in \mc{B}$, $B_i$ is nonempty
for at least one index $i$, we have $\sum_i |\mc{B}_i| \geq |\mc{B}|$.
It follows that $n \cdot bs(f^{(k)}) \geq \sum_i bs_{X_i}(f^{(k)}) \geq bs_X(f^{(k+1)})$, as required.

\iffalse
for  such that $we$ can write $B = b(B^i: i\in I)$, where $b\in\partial\mc{B}_x(f)$ and furthermore, $B^i\in \partial\mc{B}_{X^i}(f^{(k)})$ if $i\in b$ and $B^i = \emptyset$ otherwise. We denote $B^i$ by $\mathrm{Proj}_i(B)$.

 The set $\mc{B}_i:=\setcond{\mathrm{Proj}_i(B)}{B\in\mc{B}, \mathrm{Proj}_i(B)\neq \emptyset}$ forms a block packing in $\mc{B}_{X^i}(f^{(k)})$ (since the blocks $B\in\mc{B}$ are pairwise disjoint) and hence $|\mc{B}_i|\leq bs(f^{(k)})$. On the other hand, each block $B\in\mc{B}$ must satisfy $\mathrm{Proj}_i(B)\neq\emptyset$ for some $i\in I$, and hence $|\mc{B}|\leq \sum_{i\in I}|\mc{B}_i|\leq n\cdot bs(f^{(k)})$.
\fi 
\end{proof}

\begin{lemma}
\label{lemma_lim_exists}
The limit $\lim_{k\rightarrow\infty} bs(f^{(k)})^{1/k}$ exists and is finite. 
\end{lemma}

\begin{proof}
It clearly suffices to show that $\lim_{k\rightarrow\infty} \log(bs(f^{(k)}))/k$ exists and is finite. Finiteness is trivial, since $1\leq bs(f^{(k)})\leq n^k$ and hence the sequence $\log(bs(f^{(k)}))/k$ is bounded. To show that the limit exists, we use Corollary \ref{corollary_fekete}. To show that $\{\log(bs(f^{(k)}))\}_k$ satisfies the hypothesis of Corollary \ref{corollary_fekete}, it suffices to show that $bs(f^{(k+\ell)}) = \Omega(bs(f^{(k)})bs(f^{(\ell)}))$, where the constant in the $\Omega(\cdot)$ is independent of $k$ (but may depend on $n$). 
But by Lemma \ref{lemma_bs_to_rbs}, we have $bs(f^{(k+\ell)})\geq bs(f^{(k)}) \cdot M$, where $M = \min\{bs_0(f^{(\ell)}),bs_1(f^{(\ell)})\}$. By Lemma \ref{lemma_bs0vs1}, we have $M\geq bs(f^{(\ell)})/n$ and thus, it follows that $bs(f^{(k+\ell)}) \geq bs(f^{(k)})bs(f^{(\ell)})/n$ and therefore, by Corollary \ref{corollary_fekete}, we are done.
\end{proof}

Lemma \ref{lemma_lim_exists} is useful since we can now analyze the limit of an arbitrary subsequence of the sequence $\{bs(f^{(k)})^{1/k}\}_k$ that we are actually interested in. 

We now proceed to the proof of Theorem \ref{thm_bs_tensor}. We will need that $\min\{bs_0(f^{(k)}),bs_1(f^{(k)})\}\rightarrow\infty$ as $k\rightarrow\infty$. By Corollary \ref{corollary_bs_infty}, this holds whenever $\min\{bs_0(f),bs_1(f)\}\geq 2$. We now look at what happens when this is not the case. Without loss of generality assume that $bs_0(f)=1$ (since $f$ is non-monotone and hence non-constant, we have $\min\{bs_0(f),bs_1(f)\}\geq 1$). It can be checked that this happens if and only if $f$ is a conjunction of literals. Since $f$ is neither monotone nor anti-monotone, there must be at least one positive and one negative literal. In this case, it can be checked that $\min\{bs_0(f^{(2)}),bs_1(f^{(2)})\}\geq 2$. Thus, by Corollary \ref{corollary_bs_infty}, we see that $\min\{bs_0(f^{(2k)}),bs_1(f^{(2k)})\}\geq 2^k$ and by Lemma \ref{lemma_bs0vs1}, we have $\min\{bs_0(f^{(2k+1)}),bs_1(f^{(2k+1)})\}\geq bs(f^{(2k)})\geq 2^k$. It follows that $\min\{bs_0(f^{(k)}),bs_1(f^{(k)})\}\rightarrow \infty$ as $k\rightarrow\infty$. 

Let $L$ denote $\lim_{k\rightarrow \infty} bs^\ast(f^{(k)})^{1/k}$. As $bs(f^{(k)})\leq bs^\ast(f^{(k)})$ for each $k\geq 1$, we have $\lim_{k\rightarrow\infty}bs(f^{(k)})^{1/k}\leq L$. We now show that for any $\varepsilon\in (0,1)$, it is the case that $\lim_{k\rightarrow\infty}bs(f^{(k)})^{1/k}\geq L(1-\varepsilon)$. 

Fix any $\varepsilon\in (0,1)$. Let $\ell_0\in\naturals$ be chosen large enough so that $F = f^{(\ell_0)}$ satisfies the following conditions:
\begin{itemize}
\item $bs^{\ast}(F) \geq (L(1-\varepsilon/4))^{\ell_0}$, 
\item $n^{-1/\ell_0}\geq (1-\varepsilon/2)$.
\end{itemize}

We will show that $\lim_{k\rightarrow\infty} bs(F^{(k)})^{1/k\ell_0}\geq L(1-\varepsilon)$. Since $\lim_{k\rightarrow\infty}bs(f^{(k)})^{1/k}=\lim_{k\rightarrow\infty} bs(F^{(k)})^{1/k\ell_0}$, this will conclude the proof of Theorem \ref{thm_bs_tensor}.

Recall from Section \ref{sec:bs_def} that for any assignment $z$ to the variables of $F$, 
$bs^{\ast}_z(F) = \lim_{M\rightarrow\infty} bs^M_z(F)/M$. Thus, there exists an $m$ such that for any $M\geq m$, $bs^M(F)\geq M(L(1-\varepsilon/2))^{\ell_0}$. Since $\min\{bs_0(F^{(k)}),bs_1(F^{(k)})\}\rightarrow\infty$ as $k\rightarrow\infty$, there exists $k_0\in\naturals$ s.t. $\min\{bs_0(F^{(k)}),bs_1(F^{(k)})\}\geq m$ for each $k\geq k_0$. 

By Lemma \ref{lemma_bs_to_rbs}, for $k\geq k_0$, we have $bs(F^{(k+1)}) \geq bs^M(F)$, where $M = \min\{bs_0(F^{(k)}),bs_1(F^{(k)})\}$. Since $M\geq m$ by our choice of $k_0$ we know that $bs^M(F)\geq M(L(1 - \varepsilon/2))^{\ell_0}$. Moreover, by Lemma \ref{lemma_bs0vs1}, we know that $\min\{bs_0(F^{(k)}),bs_1(F^{(k)})\} = \min\{bs_0(f^{(\ell_0k)}),bs_1(f^{(\ell_0k)})\} \geq bs(F^{(k)})/n$. Thus, we have for $k\geq k_0$, $bs(F^{(k+1)}) \geq (L(1-\varepsilon/2))^{\ell_0}\cdot bs(F^{(k)})/n$. Iterating this inequality we obtain for any $k\geq k_0$,
\begin{align*}
bs(F^{(k)}) &\geq (L(1-\varepsilon/2))^{\ell_0(k-k_0)}\cdot bs(F^{(k_0)})/n^{k-k_0}\\
&\geq \frac{(L(1-\varepsilon/2))^{\ell_0k}}{C\cdot n^k}
\end{align*}
where $C>0$ is some quantity that is independent of $k$. Thus, we have
\begin{align*}
\lim_{k\rightarrow\infty} bs(F^{(k)})^{1/k\ell_0} 
&\geq \frac{(L(1-\varepsilon/2))}{n^{1/\ell_0}}\\
&\geq L(1-\varepsilon/2)^2 \geq L(1-\varepsilon)
\end{align*}
The second inequality above follows since $n^{-1/\ell_0}\geq (1-\varepsilon/2)$. Thus, we have shown that $\lim_{k\rightarrow\infty} bs(f^{(k)})^{1/k}=\lim_{k\rightarrow\infty} bs(F^{(k)})^{1/k\ell_0}\geq L(1-\varepsilon)$. Since $\varepsilon>0$ can be made arbitrarily small, this shows that $\lim_{k\rightarrow\infty} bs(f^{(k)})^{1/k}\geq L$ and concludes the proof of Theorem \ref{thm_bs_tensor}.

\subsection{Correcting a previous separation result}\label{sec:scottremark}

We use Theorem \ref{thm_bs_tensor} to correct and clarify a couple of remarks from Aaronson's paper \cite[Section 5]{aaronsonqcc}. 

Aaronson considers a function $f:\{0,1\}^6\rightarrow \{0,1\}$ due to Bublitz et al. \cite{bublitz} for the purposes of creating some separating examples. A short description of the function follows (the function is defined slightly differently by Bublitz et al.). The function $f(x_1,\ldots,x_6)$ is defined as the following depth 2 decision tree with parity gates: First compute $x_1 \oplus x_2 \oplus x_3 \oplus x_4$, if $0$ then output $x_1 \oplus x_2 \oplus x_5$, else output $x_1 \oplus x_3 \oplus x_6$.
%
%One of the nice properties of $f$ is that at every input $z \in \{0,1\}^6$, the minimal blocks in $\mc{B}_z(f)$ consist of $3$ sensitive bits, followed by $3$ blocks of size $2$ which form a triangle. This fact implies the following
It can be checked that $f$ has the following property (the proof of which is omitted): 

\begin{lemma}
\label{lemma_bublitzfn_props}
For every $z\in\{0,1\}^6$, $bs_z(f)=4$, $bs^\ast_z(f) = C^\ast_z(f)= 4.5$, and $C_z(f) = 5$.
\end{lemma}

\begin{enumerate}
\item It is claimed that $bs(f^{(k)}) = 4^k$ and $C(f^{(k)}) = 5^k$ and thus $C(f^{(k)}) = bs(f^{(k)})^{\log_4 5}$ for every 
$k\in\naturals$. However, it follows from Theorem \ref{thm_bs_tensor} that for a boolean function $g$, $\lim_{k\rightarrow\infty} (bs(g^{(k)}))^{1/k}=\lim_{k\rightarrow\infty}(bs^\ast(g^{(k)}))^{1/k}$ as $k\rightarrow \infty$, which may in general be significantly larger than $bs(f)^k$. In this case, by Lemma \ref{lemma_bublitzfn_props} and Theorem \ref{thm:m^lim}, it follows that $(bs^\ast)^{\lim}(f)= 4.5$ and hence, by Theorem \ref{thm_bs_tensor}, for any $\varepsilon > 0$ and large enough $k\in\naturals$ depending on $\varepsilon$, $bs(f^{(k)}) \geq (4.5 - \varepsilon)^k$. In particular, this example only yields $\mathrm{crit}(C,bs)\geq \log_{4.5} 5$, which is smaller than the $\log_4 5$ separation claimed.

\item It is also claimed that the family $f^{(k)}$ yields polynomial separations between the block sensitivity $bs(\cdot)$ and $RC(\cdot)$, where $RC(F)$ for any boolean function $F$ is the \emph{randomized certificate complexity of $f$} (see Section \ref{sec:CstarvsRC}). However, by Theorem \ref{thm_bs_tensor}, it follows that such an approach (irrespective of the base function $f$) can never yield a polynomial gap between $bs(\cdot)$ and $RC(\cdot)$, since
\[
bs^{lim}(f) = (C^\ast)^{\lim}(f) = \lim_{k\rightarrow \infty} (RC(f^{(k)}))^{1/k}
\] 
where the last equality follows from Claim \ref{claim_RCvsCfrac}.
\end{enumerate}

\section{Separating examples} \label{sec:examples}

In this section we prove a tight lower bound of 2 on the critical exponent for $C(f)$ and $bs^*(f)$ (and the
same tight lower bound holds for the critical exponent for $C(f)$ and $bs^*(f)$.)  We exhibit two
different families of boolean functions that attain this separation.
We also exhibit a family of boolean functions that proves a lower bound of 3/2 on the critical exponent of $bs^*(f)$ and $bs(f)$.

One of our examples uses iterated composition.  The other two examples
are obtained by composing the $n$-bit OR function $OR_n$ with a suitable function $g$.  We will need the
following simple fact:

\begin{proposition}
\label{OR composition}
Let $g$ be a non-constant boolean function and $f=OR_n \circ g$.  Then for complexity measure $m \in \{C,bs,bs^*\}$
we have:

\begin{eqnarray*}
m_1(f)&=&m_1(g)\\
m_0(f)& = & n \cdot m_0(g).
\end{eqnarray*}

\end{proposition} 

\begin{proof}
Let $I$ be the index set for the variables of $g$, so $J=[n] \times I$ is the index set for
the variables of $f$.  For $i \in [n]$, write $J_i$ for the index subset $\{i\} \times I$.

First we show $m_1(f) = m_1(g)$.
The function
$g$ is a subfunction of $f$ (i.e., can be obtained from $f$ by restricting some variables)
so $m_1(f) \geq m_1(g)$ for each of the above complexity measures $m$.  
For the reverse inequality, we argue that $C_1(f) \leq C_1(g)$, the
argument for the other two measures is similar.  Let $\alpha \in g^{-1}(1)$ be an input
for which $C_{\alpha}(g)$ is maximum.  Construct an input $\beta$ for $f$
by fixing the variables in $J_n$ according to $\alpha$ and for each $i \in [n-1]$
fix the variables in $J_i$ to some  input $y$ for $g$ such that $g(y)=0$.
It is easy to check that  $C_1(f) \leq C_{\beta}(f) = C_{\alpha}(g) =C_1(g)$.

Next we show that $m_0(f) = n \cdot m_0(g)$.   For this, write an assignment to the variables of  $f$
as $\alpha^1,\ldots,\alpha^n$ where each $\alpha^i$ is an assignment to the variables of $g$.  
We have $f(\alpha^1,\ldots,\alpha^n)=0$
if and only if $g(\alpha^1)=\cdots = g(\alpha^n)=0$.  It is easy to check that
for each of the measures $m$ under consideration, if $g(\alpha^1)= \cdots = g(\alpha^n)=0$ then
$m_{\alpha^1,\ldots,\alpha^n}(f)=m_{\alpha^1}(g) + \cdots + m_{\alpha^n}(g)$.  Thus
an input in $f^{-1}(0)$ that maximizes $m_{\alpha^1,\ldots,\alpha^n}(f)$ is one
for which $\alpha^1=\cdots = \alpha^n=\alpha$, where $\alpha$ satisfies $m_0(g)=m_{\alpha}(g)$.
This gives $m_0(f)=n \cdot m_0(g)$. 
 
\end{proof} 
    
\subsection{Achieving quadratic separation between $C(f)$ and $bs(f)$}
  \subsubsection{A Probabilistic Construction}\label{sec:example1}

In this section we construct a sequence of $n$-variate functions $g_n$ (for $n$ sufficiently large) such that $C_0(g_n)=\Omega(n)$
and $bs_0(g_n)=O(1)$.    We then define $f_n=OR_n \circ g_n$.  By Proposition
\ref{OR composition}, we have $C(f_n) \geq C_0(f_n) = n \cdot C_0(g_n) = \Omega(n^2)$,
while $bs(f_n) \leq bs^*(f_n) \leq \max(bs^*_0(f_n), bs^*_1(f_n)) \leq \max(n bs^*_0(g_n),bs^*_1(g_n)) = O(n)$.
   
This will prove:

   \begin{theorem}
  \label{thm_ex1}
    For every $n\in\naturals$ sufficiently large, there is a function $f:\{0,1\}^{n^2}\rightarrow\{0,1\}$ such that $bs(f) \leq bs^\ast(f) = O(n)$ and $C(f) = \Omega(n^2)$.
  \end{theorem}

Let us write  $\delta(x,y)$ to denote the Hamming distance between $x,y \in \{0,1\}^n$.   
   We define $g=g_n:\{0,1\}^n \rightarrow \{0,1\}$ as follows 
(we view $n$ as being sufficiently large). Choose $x_1,\dots,x_N \in \{0,1\}^n$ uniformly at random (with replacement) with $N = 2^{n/50}$. We set $g(x_i) = 1$ for each $i$, and $g(x) = 0$ otherwise.
   \begin{claim}
    With high probability, for all $i,j$ distinct $\delta(x_i,x_j) \geq \frac{n}{100}$.
   \end{claim}
   \begin{proof}
    Let $A_{i,j}$ denote the event $\delta(x_i,x_j) < \frac{n}{100}$. Let $x$ be a fixed point in $\{0,1\}^n$ and $B(x,r)$ denote the Hamming ball of radius $r$ and center $x$. Then $|B(x, r)| = \sum\limits_{i=0}^{r} \binom{n}{i}.$ Thus we have
    \[B\left(x,\frac{n}{100}\right) < 2\binom{n}{n/100} \leq 2{(100e)}^{n/100} < 2^{n/10}.\]
    These inequalities imply that
    \[\mathbf{P}(A_{i,j}) = \frac{B\left(x,\frac{n}{100}\right)}{2^{n}} < 2^{-9n/10}.\]
    By the union bound the hypothesis fails with probability at most
    \[2^{-9n/10} \binom{N}{2} = o(1).\]
   \end{proof}
   If the hypothesis of the claim holds and $g(x) = 0$, then all but possibly one of the blocks for $g$ at $x$ will have size at least $\frac{n}{200}$. Thus, at most $200$ blocks can be packed and $bs_0(g) \leq 200$. Likewise, this bound on the size of blocks implies that $bs^{\ast}_0(g) \leq 200$.
   
   We now argue that all sufficiently large subcubes of $\{0,1\}^n$ will contain a $1$ of $g$ almost surely. 
   
   \begin{claim}
   With high probability, $C_0(g) \geq \frac{n}{100}$
   \end{claim}
   \begin{proof}
   Its enough to show that every subcube of co-dimension $\frac{n}{100}$ will contain a $y$ such that $g(y) = 1$. For each $S$ which is a subcube of co-dimension $\frac{n}{100}$, denote $A_S$ as the event $g(x) = 0$ for all $x \in S$. Then
    \[ \mathbf{P} ( A_S) \leq (1 - 2^{-n/100})^{N} < \text{exp}(-\frac{N}{2^{n/100}}) = \text{exp}(-2^{n/100}) \]
    There are $\binom{n}{n/100} 2^{n/100} < 2^{2n}$ subcubes of co-dimension $\frac{n}{100}$. Thus by union bound the hypothesis fails with probability at most \[ \text{exp}(-2^{n/100}) 2^{2n} = o(1).\]
    
    \end{proof}
    We have shown, for sufficiently large $n$, that with high probability a random function $g$ satisfies  $bs_0^{\ast}(g) \leq 200$ and $C_0(g) \geq \frac{n}{100}$. Thus for each $n$ sufficiently large, there exists a function $g_n$ with this property.
    
    \iffalse
    Fix a $g$ with this property. Taking $f:\{0,1\}^{n^2} \rightarrow \{0,1\}$ to be an OR of $n$ copies of $g$ and applying Proposition \ref{OR composition}, we conclude that \[bs(f) \leq bs^{\ast}(f) \leq \max(bs^{*}_0(f), bs^*_1(f)) \leq \max(200n, n) \leq 200n\]
     and \[C(f) \geq C_0(f) \geq \frac{n^2}{100}.\]
     \fi
  
  \subsubsection{A Construction Using Iterated Composition}\label{sec:example2}

    In this section we construct a function $f$ on $n$ variables for which $C^{lim}(f) \geq \frac{n}{2}$ and $(C^*)^{\lim}(f) \leq 4 \sqrt{n}$. For any $\epsilon >0$, we may choose $n$ large enough to conclude that crit$(C^*,C) \geq 2-\epsilon$.
    %We can choose $n$ large enough such that when we compose $f$ with itself, the function $f^k$ will in the limit exhibit a near quadratic separation between $C$ and $C^*$. This then implies that crit$(C,C^*) \geq 2-\epsilon$ for any $\epsilon$. It achieves the same separation as the probabilistic construction described earlier, however we include it as it is very different in spirit and illustrates how one might use Theorem \ref{thm:m^lim} in practice.
    
   Let $d,k,n$ be positive integers such that $n \geq k \geq d$, $d \mid k$, and $k \mid n$. We define $f:\{0,1\}^n \rightarrow \{0,1\}$ to be the following boolean function on $n$ variables:
    
View the $n$ indices of the input $x$ as being divided into $\frac{n}{k}$ disjoint groups, with each group containing $k$ indices. $f$ accepts if and only if $|x| \geq d$ and all the 1's in $x$ can be found in a single group. Note that $f(x) = 1$ implies $|x| \leq k$. %It can be easily checked that $C_0(f) = \frac{(k-d+1)n}{k}$, obtained when $x = (0,0,\dots,0)$, and $C_1(f) = n-k+d$. Also $bs_0(f) = \frac{n}{d}$, $bs_1(f) = n-k+d$.

Although $f$ shows no separation between $bs(f)$ and $C(f)$, the key is that both the zero and one certificate complexity for $f$ are large, while the zero block sensitivity is small. Also, any $1$-assignment for $f$ contains many 0 indices.

 In the following analysis, we assume $n$ is an even perfect square and set $k:= 2\sqrt{n}$ and $d:=\sqrt{n}$. We wish to bound $C^{lim}(f)$ and $(C^*)^{\lim}(f)$. By Theorem \ref{thm:m^lim}, it is enough to bound $\widehat{C}(f)$ and $\widehat{C^*}(f)$ instead.
 \begin{claim} \label{claim2}
 For the boolean function $f:\{0,1\}^n \rightarrow \{0,1\}$ defined above we have:
 \[\widehat{C^*}(f) \leq 4 \sqrt{n}.\]
 \end{claim}
 \begin{proof}
  %Recall that for $m \in \{s,C^*,C\}$ the complexity measure $\hat{m}(f)$ is defined as the maximum over input selectors $\vec{\alpha}$ of the minimum over $A \in \mc{M}_{\vec{\alpha}}(f)$ of $\rho(A)$. Each profile matrix $A$ is of the form 
  %\[
  %\begin{bmatrix}
  %    p_{\alpha^0}(w_0) \\
  %    p_{\alpha^1}(w_1)
  %  \end{bmatrix}
  % \]
 %where $\alpha^0, \alpha^1$ are zero and one inputs respectively, and $w_0, w_1$ are fractional hitting sets at $\alpha^0, \alpha^1$ respectively. 
 We proceed by showing that for any assignment selector $\vec{\alpha} = (\alpha^0,\alpha^1)$, we can find a pair of hitting sets $(w^0,w^1)$ such that the corresponding profile matrix has all eigenvalues less than $4 \sqrt{n}$. We look at the 0 assignments first, and for each possible $\alpha^0$ we exhibit a small fractional hitting set $w_0$.
    
 \textbf{Case 1}, $\alpha^0 = (0,0,\ldots,0)$:
 
 Here we choose $w_0:=(\frac{1}{d},\frac{1}{d},\cdots, \frac{1}{d})$. It follows that, $w_0$ is a fractional hitting set as each block for this assignment has size at least $d$. For this hitting set the profile vector 
 $p_{\alpha^0}(w_0) = (\frac{n}{d},0)$.
 
 \textbf{Case 2}, $|\alpha^0| = j$, and all 1's in $\alpha^0$ appear in the same group:
 
 Note this means that $j < d$ as $\alpha^0$ is a 0 assignment. Let $X_1$ be the set of indices for $\alpha^0$ which are 1's, let $G_1$ be the group which contains $X_1$. Pick an $s \in X_1$, we define a fractional hitting set $w_0$ to assign weight $1$ to $s$,  weight $1$ to all indices in $G_1 \setminus X_1$, and weight 0 otherwise. To see that $w_0$ is indeed a hitting set, note that if $B$ is a block for $\alpha^0$, then either $B \subseteq G_1$ or $X_1 \subset B$. If $X_1 \subset B$, then $s \in B$ and it has been assigned weight 1. If $B \subseteq G_1$ then $B$ must contain a 0 index in $G_1$ as $|\alpha^0| <d$, this index was assigned weight 1 by $w_0$. Thus $w_0$ is a hitting set and the profile vector $p_{\alpha^0}(w_0) = (k-j,1) \leq (k, 1)$.

 \textbf{Case 3}, At least two different groups in $\alpha^0$ contain 1's:
 
 Let $G_1, G_2$ be two distinct groups containing 1's. Let $X_1, X_2$ be the set of indices which are assigned 1 by $\alpha^0$ in $G_1, G_2$ respectively. Then if $B$ is a block for $\alpha^0$, either $X_1 \subseteq B$ or $X_2 \subseteq B$. We define $w_0$ to assign weight 1 to an index in $X_1$ and in index in $X_2$. This will be a hitting set, and the profile vector $p_{\alpha^0}(w_0) = (0,2)$. This concludes the analysis of each possible 0 assignment.

 \textbf{The 1 assignments $\alpha^1$:}

 If $\alpha^1$ is a 1 assignment then $|\alpha^1| \geq d$ and all the 1's appear in a single group, call it $G_1$. In this case we define $w_1$ to assign weight 1 to all indices outside $G_1$, and weight 1 to $d$ indices in $G_1$ which are assigned 1 by $\alpha^1$. This will be a hitting set as any block must contain a 0 index outside of $G_1$ or leave less than $d$ 1's inside of $G_1$ after flipping the indices in $B$. Here the profile vector $p_{\alpha^1}(w_1) = (n-k,d)$.
 
  If $M, M'$ are $2 \times 2$ matrices with nonnegative entries, and $M \leq M'$ entry by entry, then $\rho(M) \leq \rho(M')$. Considering this along with the 3 cases of 0 assignments above, bounding $\widehat{C^*}(f)$ reduces to bounding the largest eigenvalues of the following matrices:
  \[  \begin{bmatrix}
      \frac{n}{d} & 0 \\
      n-k & d
    \end{bmatrix} \ \ \ \
     \begin{bmatrix}
      k & 1 \\
      n-k & d
    \end{bmatrix} \ \ \ \
     \begin{bmatrix}
      0 & 2 \\
      n-k & d
    \end{bmatrix}
    \]
    Here the second matrix has the largest eigenvalue of the three. It is easy to check that $k= 2\sqrt{n}$, $d = \sqrt{n}$ implies its largest eigenvalue is less than $4 \sqrt{n}$.
 \end{proof}
 
 \begin{claim} 
 \[\widehat{C}(f) \geq \frac{n}{2}.\]
 \end{claim}
 \begin{proof}
   To prove this we choose an assignment selector $\vec{\alpha}$ for which all profile matrices $A \in \mc{M}_{\vec{\alpha}}(f)$ have an eigenvalue larger than $\frac{n}{2}$. We set $\alpha^0 := (0,0,\cdots,0)$ and $\alpha^1$ to have exactly $d$ 1's in the first group, and be identically 0 in every other group.
   
   Any certificate for $\alpha^0$ must fix $k-d+1$ indices in each group, thus must fix $\frac{n}{k}(k-d+1)$ in total. It follows that any minimum certificate $w_0$ (viewed as a boolean valued weight function) yields the profile vector $p_{\alpha^0}(w_0) = (\frac{n}{k}(k-d+1),0)$.
   
   Likewise, any certificate for $\alpha^1$ must fix all 1 indices (there are $d$ of them), and fix all the 0 indices outside the unique group containing the 1's. Thus any minimal profile vector $p_{\alpha^1}(w_1) = (d, n-k)$. The claim then reduces to looking at the maximum eigenvalue of the matrix 
   \[A = \begin{bmatrix}
       \frac{n}{k}(k-d+1) & 0 \\
       n-k & d        
     \end{bmatrix}.\]
    When $k = 2\sqrt{n}$ and $d = \sqrt{n}$ this matrix has an eigenvalue larger than $\frac{n}{2}$.

   \end{proof}

\subsection{A separation between fractional block sensitivity and block sensitivity}\label{sec:example3}

\begin{theorem}
\label{thm_bssep}
For infinitely many natural numbers $n$, there is an $n^2$-variate 
function $f_n:\{0,1\}^{n^2}\rightarrow \{0,1\}$ s.t. $bs(f_n) = O(n)$ and $bs^{\ast}(f_n) = \Omega(n^{3/2})$.
Therefore $\crit(bs^*,bs) \geq 3/2$.
\end{theorem}

To construct $f_n$, we build an $n$-variate function $g=g_n:\{0,1\}^{n}\rightarrow\{0,1\}$ 
satisfying $bs_0(g) = O(1)$ and $bs_0^*(g) = \Omega(\sqrt{n})$. 
We then define $f_n=OR_n \circ g$.
Using Proposition \ref{OR composition}  we conclude that $bs(f_n) = O(n)$ and $bs^*(g) = \Omega(n^{3/2})$.
(In a previous version of this paper our construction for the function $g$ was random, and gave a  weaker bound of 
$\Omega(\sqrt{n/\log n})$ for $bs_0^*(g)$, which was still enough to show $crit(bs^*,bs)  \geq 3/2$.
Avishay Tal (personal communication) gave an alternate explicit construction
which gave the bound of the theorem, which is what we present here.)
 
The function $g=g_n$ is defined for any $n$ of the form $\binom{s}{2}$ for an integer $s$.  Identify the input
bits of $g$ with the edges of the complete graph $K_s$.  An assignment $\alpha$ to the variables of $g$ can be viewed as an
undirected graph $G_{\alpha}$ consisting of those edges assigned 1 by $\alpha$. 
We denote by $H_i$ the star centered at vertex $i$ and by $x^i$ the corresponding input in $\{0,1\}^{I_1}$. The function $g(x)$ is defined to be $1$ iff $x=x^i$ for some $i\in [n]$. 

We now show that $g$ satisfies:

\begin{enumerate}[(a)]
\item $g(0^{n}) = 0$,
\item $bs^\ast_{0^{n}}(g) \geq s/2 = \Theta(\sqrt{n})$,
\item $bs_0(g)  \leq 3$.
\end{enumerate}

Property (a)  is immediate. For property (b),  note that the blocks for $g$ at $0^s$ are the stars $H_i$, and each edge appears
in exactly 2 of these stars, so putting weight $1/2$ on each of these stars gives a fractional
packing of blocks of total weight $s/2$.

We now prove property (c). Fix any assignment $a\in g^{-1}(0)$ and let $G_a$ denote the corresponding graph. We show that $bs_a(g) \leq 3$. Assume, for the sake of contradiction, that $bs_a(g) \geq 4$. Then, there exists four edge-disjoint
graphs $J_1, J_2, J_3$, and $J_4$ such that starting from $a$ and flipping all the bits indexed by $J_\ell$ (for any $\ell\in [4]$) produces one
of the graphs $H_i$.  
By renaming input bits if necessary, we may assume that the star graphs thus produced are $H_1, H_2, H_3$, and $H_4$ respectively.   Thus $J_i = G_a \oplus H_i$, where $\oplus$ denotes symmetric difference.  Since edge $\{1,2\}$ belongs to $H_1$ and $H_2$ we must have $\{1,2\} \in G_a$
so that $J_1$ and $J_2$ are disjoint.  But then $\{1,2\} \in J_3 \cap J_4$, contradicting their disjointness.

\section{Acknowledgements}
We would like to thank Avishay Tal for his permission to include his separating example used
in the proof of Theorem \ref{thm_bssep}.

\bibliographystyle{alpha}
\bibliography{tensor}

\appendix

\section{Fractional Certificate complexity vs. Randomized Certificate complexity}
\label{sec:CstarvsRC}
In \cite{aaronsonqcc}, Aaronson introduced the notion of the \emph{Randomized Certificate complexity} of a boolean function $g:\{0,1\}^n\rightarrow\{0,1\}$. 

For $g$ as above and an input $z\in \{0,1\}^n$, a \emph{Randomized Verifier} for $z$ is a non-adaptive randomized query algorithm that expects an $n$-bit input $z'$ and behaves as follows.\footnote{Strictly speaking, this corresponds to the definition of \emph{non-adaptive} Randomized Certificate complexity from Aaronson's paper. However, by Lemma $2.1$ of \cite{aaronsonqcc}, it follows that this is within a fixed universal constant of the Randomized Certificate complexity of $f$.} The query algorithm queries each bit $z'_i$ of its input independently with some fixed probability $\lambda_i\in [0,1]$ and accepts iff it finds no disagreement between $z'$ and $z$. Moreover, the query algorithm satisfies the following soundness property: given any $z'$ s.t. $g(z')\neq g(z)$, it rejects $z'$ with probability at least $1/2$.

The cost of such a verifier is the expected number of bits of $z'$ that are queried, which is $\sum_{i\in [n]} \lambda_i$. The Randomized Certificate complexity of $f$ at $z$ is defined to be $RC^z(f) :=\min\setcond{c}{\text{There is a cost $c$ verifier for $z$}}$. The Randomized Certificate complexity of $f$ is defined to be $RC(f) := \max_{z\in \{0,1\}^n}RC^z(f)$.

The following relation between $RC(g)$ and $C^\ast(g)$ can be proved.

\begin{claim}
\label{claim_RCvsCfrac}
Fix any boolean function $g:\{0,1\}^n\rightarrow \{0,1\}$ and any $z\in \{0,1\}^n$. Then, $RC^z(g) = \Theta(C^\ast_z(g))$. That is, the quantities $RC^z(g)$ and $C^\ast_z(g)$ are within a fixed universal constant factor for any $g,z$ as above. In particular, $RC(g) = \Theta(C^\ast(g))$.
\end{claim}

\begin{proof}
We first show that $C^{\ast}_z(g) = O(RC^z(g))$. Fix an optimal verifier for $z$ and let $\lambda_i\in [0,1]$ be the probability that it queries the $i$th bit of its input $z'$. Fix any block $B\in\mc{B}_z(g)$ and consider the input $z'$ obtained by starting with $z$ and flipping the bits indexed by $B$. Since $g(z')\neq g(z)$, the soundness property of the verifier implies that the verifier probes a bit in $B$ is at least $1/2$; in particular, by the union bound, $\sum_{i\in B}\lambda_i \geq 1/2$. 

Consider the function $\sigma:[n]\rightarrow [0,1]$ defined by $\sigma(i):= \min\{2\lambda_i, 1\}$. The above immediately implies that for any block $B\in\mc{B}_z(g)$, we have $\sum_{i\in B}\sigma(i)\geq 1$ and hence $\sigma\in\mc{W}^*_z(g)$. Moreover, $|\sigma|\leq 2\sum_i \lambda_i = O(RC^z(g))$ since we considered an optimal verifier for $z$. Thus, $C^\ast_z(g) \leq |\sigma|= O(RC^z(g))$.

We now show that $RC^z(g) = O(C^\ast_z(g))$. Fix an optimal fractional certificate $\sigma\in\mc{W}^\ast_z(g)$. Consider the randomized query algorithm that queries each bit of its $n$-bit input $z'$ with probability $\sigma(i)$ and rejects on finding any disagreement with $z$. To show that this gives us a verifier for $z$, we need to verify the soundness property. Given any input $z'$ s.t. $g(z')\neq g(z)$, the set of indices $B\subseteq [n]$ where $z$ and $z'$ differ is a block of $g$ at $z$ and hence, we must have $\sum_{i\in B} \sigma(i) \geq 1$. 

Thus, the probability that the verifier accepts $z'$ is equal to $\prod_{i\in B}(1-\sigma(i))\leq \exp\{-\sum_{i\in B}\sigma(i)\} \leq e^{-1}< 1/2$. This proves the soundness property of the verifier. Note that the expected number of queries made by the verifier is exactly $|\sigma| = C^{\ast}_z(g)$ and hence, $RC^z(g) \leq C^{\ast}_z(g)$.
\end{proof}



%\bibliographystyle{IEEEtran}
%\begin{thebibliography}{1}
%\end{thebibliography}

\end{document}